\documentclass[conference]{IEEEtran}
\IEEEoverridecommandlockouts

\usepackage{amsmath}
\usepackage{amsfonts}
\usepackage{amsthm}
\usepackage{amssymb}

\usepackage{cite}
\usepackage{booktabs}
\usepackage{graphicx}
\usepackage{textcomp}
\usepackage{xcolor}
\usepackage{url}

\usepackage{fontawesome5}

\usepackage{algorithm}
\usepackage[noend]{algpseudocode}
\usepackage{listings}
\usepackage{array}

\usepackage{wasysym}
\usepackage{tikz}
\usepackage{makecell}
\usepackage{pifont}
\usepackage{threeparttable}
\usepackage{enumerate}
\usepackage{enumitem}

\usepackage{soul}
\sethlcolor{yellow}

\newtheorem{proposition}{Proposition}
\theoremstyle{definition}

\newcommand{\para}[1]{\noindent {\bf #1} \hspace{2pt}}

\newcommand*\circledsmall[1]{%
  \tikz[baseline=(char.base)]{
    \node[
      circle,
      draw,
      thick,
      inner sep=0.2pt,
      font=\fontsize{7.5}{8}\selectfont
    ] (char) {\strut #1};
  }
}

\definecolor{codered}{rgb}{0.6,0,0}
\definecolor{codegray}{rgb}{0.4,0.4,0.4}

\lstdefinestyle{sqlstyle}{
    language=SQL,
    basicstyle=\scriptsize\ttfamily,
    keywordstyle=\color{codered}\bfseries,
    commentstyle=\color{codegray},
    showstringspaces=false,
    framerule=1pt,
    rulecolor=\color{black},
    breaklines=true,
    breakatwhitespace=true,
    numbers=none,
    keepspaces=true,
    escapeinside={(*@}{@*)},
    moredelim=**[s][\hl]{--@hlstart}{--@hlend}
}

\begin{document}

\title{A Set-Theoretic Approach to Detecting Logic Bugs in DBMS Inner Join Optimizations}

\author{
\IEEEauthorblockN{Ce Lyu$^{1}$, Changzheng Wei$^{1,2}$, Yanhao Wang$^{1}$, Jie Liang$^{3}$, Li Lin$^{2}$,\\
Hanghang Wu$^{2}$, Minghao Zhao$^{1*}$\thanks{$^*$Corresponding authors.}, Ying Yan$^{2*}$, Aoying Zhou$^{1}$}
\IEEEauthorblockA{
$^{1}${\em East China Normal University} \quad $^{2}${\em Digital Technologies, Ant Group} \quad $^{3}${\em Tsinghua University}\\
51275903097@stu.ecnu.edu.cn, changzheng.wcz@antgroup.com, yhwang@dase.ecnu.edu.cn, liangjie.mailbox.cn@gmail.com\\
\{felix.ll, hanghang.whh\}@antgroup.com, mhzhao@dase.ecnu.edu.cn, fuying.yy@antgroup.com, ayzhou@dase.ecnu.edu.cn}
}

\maketitle

\begin{abstract}
The query optimizer is a fundamental component of database management systems that determines the most efficient execution strategy for a given query by evaluating alternative query plans. Among its tasks, join optimization plays a central role, as the order of joins in multi-table queries can significantly affect execution performance. However, due to the inherent complexity of join optimization, logical bugs are inevitable and often difficult to detect. While existing fuzzing tools have shown notable success in uncovering crash- and performance-related errors, effectively identifying logical bugs---cases in which the system produces incorrect query results---remains largely unresolved.

In this paper, we propose a metamorphic testing approach to detect DBMS bugs related to \texttt{INNER JOIN} optimization through the lens of set theory. For each testing case, equivalent queries are generated based on a basic set operation---\emph{intersection}---and three semantics-preserving transformation rules, i.e., symmetric join transformation, asymmetric difference transformation, and symmetric difference transformation, are introduced. These rules rewrite a simple \texttt{NATURAL/INNER JOIN} query into a more complex, yet semantically equivalent, form. We implement this design in JoinEquiv, which serves as a testing oracle to systematically uncover logical inconsistencies in DBMS query processing by comparing the results of original and transformed queries. Using JoinEquiv, we uncovered 29 previously unknown issues in mainstream DBMSs (MySQL, TiDB, DuckDB, and Percona), and 27 of them were officially confirmed. JoinEquiv reveals deep logical flaws in DBMS optimizers and executors, underscoring its value in enhancing DBMS robustness.
\end{abstract}

\begin{IEEEkeywords}
Database testing, query optimization, join equivalence, logical bugs
\end{IEEEkeywords}

\section{Introduction}
\label{sec:intro}

Database management systems (DBMSs) constitute the backbone of modern computing infrastructure, enabling efficient storage, management, and retrieval of large-scale data \cite{chamberlin1976relational, stonebraker2013intel, stonebraker2024goes, hai2025quantum}.
Within a DBMS, the query optimizer is a critical component that determines the most efficient execution plan for a given query \cite{ding2018plan, ryan2019neo, zhang2023simple, neumann2018adaptive, xin24spatial}.
Join optimization, in particular, is central to query processing, as joins often dominate the cost of complex analytical workloads \cite{zhao2025debunking,birler2024robust}.
Optimizing joins is notoriously challenging due to the factorial growth of the join search space and the uncertainty of cost estimation.
Consequently, query optimizers employ sophisticated algorithms to determine optimal join orders \cite{hu2025output,kalumin2025optimizing}.
While effective, these techniques render DBMS query optimizers highly complex, making bugs inevitable.

Among these bugs, some are relatively easy to detect, while others are not.
Specifically, {\em crash bugs} cause the system to terminate unexpectedly or become unavailable, often arising from unhandled exceptions, resource conflicts, or internal errors in storage or indexing modules; {\em performance bugs} do not compromise correctness but significantly degrade query efficiency, resulting in long execution times, high resource consumption, or reduced throughput.
Both types of bugs are relatively easy to detect, as their effects are directly observable.
In contrast, {\em logic bugs} produce incorrect query results (i.e., violating the intended semantics of the workload) but do not cause the system to fail.
They are particularly difficult to detect, as their manifestations are subtle and may arise only under specific data distributions or complex query patterns~\cite{wikipedia_test_oracle_2025, howden2006theoretical}.
While state-of-the-art fuzzers (e.g., SQLsmith~\cite{seltenreich2025sqlsmith} and AFL~\cite{afl2025}) efficiently uncover crash and performance defects, they generally lack the ability to detect logic bugs.

To address the challenge of detecting logic bugs, several approaches have been proposed.
Differential testing, for example, identifies potential bugs by comparing the results of the same query across different DBMSs or versions of the same DBMS.
While effective at detecting inconsistencies between implementations, its applicability is often limited to the common core of SQL functionalities~\cite{slutz1998massive}, as different DBMSs vary in their support for the SQL standard and frequently provide proprietary extensions.
Another approach is pivoted query synthesis (PQS)~\cite{pqs}, which validates query results by constructing auxiliary queries centered on a specific pivot row.
Although PQS can be effective in certain scenarios, it requires substantial engineering effort and is less capable of capturing set-level inconsistencies, such as duplicates or missing tuples.

Recently, metamorphic tests have been increasingly used for DBMS testing~\cite{tlp, thanos, norec, song2025detecting}.
Representative frameworks include ternary logic partitioning (TLP)~\cite{tlp} and non-optimizing reference engine construction (NoREC)~\cite{norec}.
These approaches transform a query into semantically equivalent variants, either through predicate partitioning or by bypassing the optimizer's execution path, to check for inconsistencies.
Transformed query synthesis (TQS) \cite{tang2023detecting} utilizes optimization hints to verify query results against a ground truth.
Recently, differential query planning (DQP) \cite{dqp} adopts a similar hint-based strategy but focuses on detecting inconsistencies between different physical plans for the same query.
We aim to present a metamorphic testing scheme that efficiently and effectively detects logic bugs in DBMS join optimization.
Our key insight is that {\em an \texttt{INNER/NATURAL JOIN}) can be semantically expressed as a combination of multiple set intersection operations.}
Such a theoretical equivalence should be strictly upheld by any standards-compliant DBMS.
However, we discovered that this is often not the case.
A highly illustrative motivating example is shown in Listing~\ref{lst:motivating-example}, which depicts a bug we discovered in TiDB \cite{tidb}.
We crafted a simple \texttt{INNER JOIN} query whose \texttt{ON} predicate, \texttt{(-1.35703815E8 >= t0.c0)}, is provably \texttt{FALSE}.
A fundamental principle of SQL dictates that such a query can produce no joined rows, and thus its final result set must be empty.
In blatant contradiction to this principle, TiDB returns a non-empty result $\{0\}$.
As shown in Listing~\ref{lst:cross-plan}, the root cause became evident upon inspecting the query's execution plan.
Astonishingly, the \texttt{HashJoin} operator was labeled as a \texttt{CARTESIAN inner join}, which signifies a cross product.
This provides direct evidence that the query optimizer had completely discarded the essential \texttt{ON} predicate.
By erroneously demoting the \texttt{INNER JOIN} to a \texttt{CROSS JOIN}, the optimizer exhibited not a subtle semantic misinterpretation but a catastrophic failure in its core logic.

\begin{lstlisting}[
  style=sqlstyle,
  caption={An exemple of a logic bug where an \texttt{INNER JOIN} is incorrectly transformed into a \texttt{CARTESIAN INNER JOIN} in TiDB.},
  label={lst:motivating-example}
]

CREATE TABLE t0(c0 BOOL UNSIGNED NOT NULL);
CREATE TABLE t1 LIKE t0;
INSERT INTO t0(c0) VALUES (true);
INSERT INTO t1(c0) VALUES (false);
-- (*@\colorbox{yellow!65}{\color{codegray}Original Query}@*)
SELECT DISTINCT t1.c0 FROM t1 INNER JOIN t0 ON -1.35703815E8>=t0.c0 WHERE t0.c0; 
-- Expected Result: empty set | Actual Result: {0}. (*@\bugicon@*)
-- (*@\colorbox{yellow!65}{\color{codegray}Equivalent Query}@*)
(SELECT DISTINCT t1.c0 FROM t1 LEFT JOIN t0 ON -1.35703815E8>=t0.c0 WHERE t0.c0)
 INTERSECT
(SELECT DISTINCT t1.c0 FROM t1 RIGHT JOIN t0 ON -1.35703815E8>=t0.c0 WHERE t0.c0); 
-- Expected Result: empty set | Actual Result: empty set.(*@\checkicon@*)
\end{lstlisting}

\begin{lstlisting}[
  language=SQL,
  caption={The anomalous execution plan in TiDB, which reveals a destructive rewrite where the \texttt{ON} condition is discarded.},
  label={lst:cross-plan}
]

-> HashAgg ... group by: t1.c0 ... (rows=1)
    -> (*@\colorbox{red!30}{{\color{codered}\bfseries HashJoin} ... {\color{codered}\bfseries CARTESIAN inner join} ... ({\color{codered}\bfseries rows}=1)}@*)
\end{lstlisting}

The above example reveals that such optimizer-level logic bugs often stem from incorrect handling of algebraic equivalences during query rewriting. 
Although existing studies have achieved significant results in logic bug detection, our systematic investigation reveals that these methods still exhibit notable deficiencies in supporting equivalent transformations at the relational algebra level.
Table~\ref{tab:feature_comparison} compares the feature coverage of typical testing approaches with respect to SQL set operations, including \texttt{EXCEPT}, \texttt{INTERSECT}, \texttt{UNION}, and their \texttt{ALL} variants.
As shown in Table~\ref{tab:feature_comparison}, PQS, NoREC, DQP, and TQS do not support \texttt{EXCEPT}, \texttt{INTERSECT}, or \texttt{UNION} operations, while TLP supports only \texttt{UNION} (including \texttt{ALL}).
These results reveal that existing approaches focus primarily on predicate partitioning or optimizer path validation, resulting in limited testing capability.

\begin{table}[t]
    \centering
    \caption{Comparison of SQL feature coverage of different testing approaches. Here, ``$\bigcirc$'' and ``$\times$'' indicate whether an approach supports the feature.}
    \label{tab:feature_comparison}
    \vspace{-1em}
    \begin{tabular}{cccc}
        \toprule
        \makecell{\textbf{Approach}} & \makecell{\textbf{\shortstack{EXCEPT\\(ALL)}}} & \makecell{\textbf{\shortstack{INTERSECT\\(ALL)}}} & \makecell{\textbf{\shortstack{UNION\\(ALL)}}} \\
        \midrule
        PQS      & $\times$      & $\times$     & $\times$      \\
        NoREC    & $\times$      & $\times$     & $\times$      \\
        TLP      & $\times$      & $\times$     & $\bigcirc$    \\
        DQP      & $\times$      & $\times$     & $\times$      \\
        TQS      & $\times$      & $\times$     & $\times$      \\
        JoinEquiv (*) & $\bigcirc$    & $\bigcirc$   & $\bigcirc$   \\
        \bottomrule
    \end{tabular}
\end{table}

To address this issue, we perform an in-depth exploration of set-theoretic equivalences in join queries for test case generation.
Leveraging the property that an \texttt{INNER JOIN} (or \texttt{NATURAL JOIN}) is semantically equivalent to a combination of set intersection operations, we design three transformation rules (\S\ref{sec:sjt-rule}, \S\ref{sec:adt-rule}, and \S\ref{sec:sdt-rule}).
These rules enable the construction of semantically equivalent query pairs, whose result sets can be compared to uncover deep logical inconsistencies in query rewriting or execution semantics.
Based on this equivalence-driven query generation method, we develop a metamorphic testing framework, JoinEquiv, implemented on top of \emph{SQLancer}~\cite{sqlancer} and made publicly available.\footnote{\url{https://github.com/DBFuzzing/JoinEquiv}}
Compared with existing approaches (see Table~\ref{tab:feature_comparison}), JoinEquiv achieves full support for all three set operations and their \texttt{ALL} variants.

To evaluate the effectiveness of our JoinEquiv framework, we conduct extensive tests on four production-grade DBMSs: MySQL, TiDB, DuckDB, and Percona.
JoinEquiv reported 29 previously unknown logical inconsistencies.
Among these, 27 have already been confirmed by developers, 14 of which involve critical bugs in query optimizer rewrite logic or executor semantic handling.
Cross-oracular validation shows that the vast majority of the remaining bugs can only be detected by our intersection-equivalence transformations; i.e., they cannot be detected by existing metamorphic testing approaches such as TLP and DQP, thus highlighting the unique and complementary value of JoinEquiv.
We believe that the principled methodology, low implementation effort, and broad applicability of JoinEquiv will lead to its widespread adoption to improve the robustness of DBMSs.

In summary, our main contributions in this paper include:
\begin{itemize}
    \item We propose a new metamorphic testing paradigm that bridges relational algebra and set theory, mapping \texttt{INNER JOIN} to set-theoretic intersection to build principled logic-level oracles.
    \item We design a structurally complete set of minimal transformation rules (SJT, ADT, and SDT) that collectively encompass all core SQL set operators. They serve as normal forms for intersection semantics, supported by formal proofs that establish their equivalence and identify their specific validity boundaries across sets and multisets.
    \item We implement and evaluate JoinEquiv on widely used DBMSs to uncover many previously unknown logic bugs.
\end{itemize}

\section{Background}

\subsection{Set and Multiset Semantics in SQL}

Although SQL is rooted in the relational model \cite{codd1970relational}, there is a key semantic divergence between SQL implementations and Codd's original pure set-based theory: modern DBMSs commonly adopt the bag (or multiset) model, which allows for identical tuples in a data table.
To cope with such dual semantics, SQL's set operators, such as \texttt{UNION}, \texttt{INTERSECT}, and \texttt{EXCEPT}, have evolved two variants: the default \texttt{DISTINCT} mode for working with sets (de-duplication), and the \texttt{ALL} mode for working with bags (preserving duplicates).
This semantic complexity extends to join operations. Both \texttt{INNER JOIN} and \texttt{NATURAL JOIN}, in essence, can be viewed as generalized intersection operations since their core goal is to find combinations of tuples that satisfy specific conditions in two relationships.

\subsection{Query Optimization and Rewrite Logic}

The query optimizer is a critical component that translates declarative SQL into efficient physical execution plans \cite{selinger1979access, graefe1993volcano}.
This typically involves navigating a combinatorial search space of join orders, physical operators, and access paths \cite{Chaudhuri1998Overview}.
To manage such complexity, optimizers rely heavily on heuristic rewrite rules (e.g., predicate pushdown and join reordering) and on cardinality estimation \cite{leis2018query}.
However, ensuring the correctness of these equivalent transformations is notoriously difficult.
The interplay between complex rewrite logic and diverse SQL features makes the optimizer a frequent source of logic bugs, in which a flawed transformation causes the DBMS to violate query semantics and produce incorrect results.

\begin{figure*}[t]
  \centering
  \includegraphics[width=.9\linewidth]{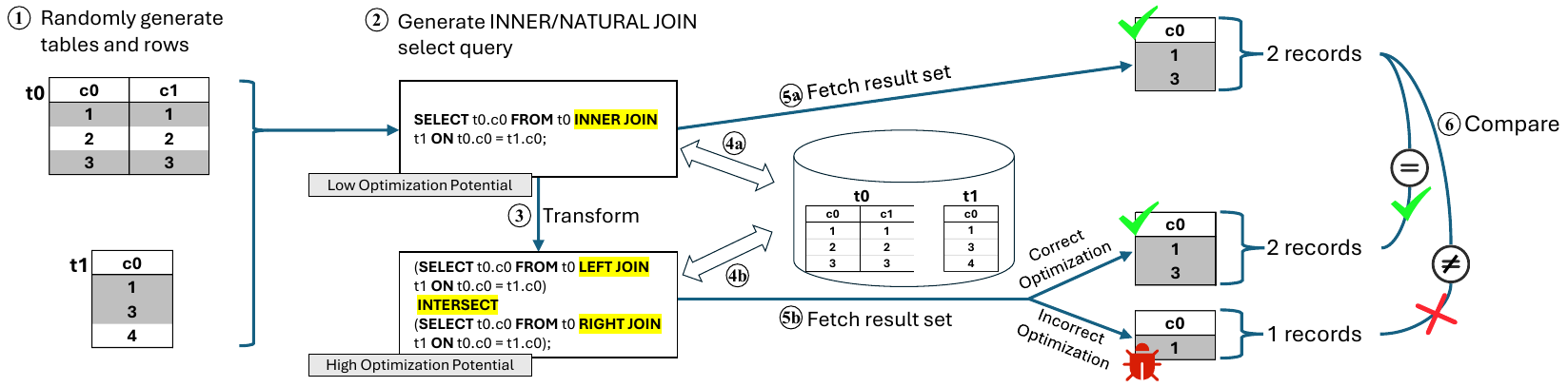}
  \vspace{-1em}
  \caption{Illustration of the procedure of JoinEquiv. Specifically, it translates an \texttt{INNER/NATURAL JOIN} query into a semantically equivalent but structurally more complex query, enabling the detection of logic bugs by comparing the result sets of the two queries. Here, the top result set is the output of the original query, while the bottom two respectively illustrate the correct and incorrect outputs of the transformed query.\label{fig:overview}}
\end{figure*}

\subsection{Metamorphic Testing}

Metamorphic testing is an effective methodology in software testing \cite{chen1998metamorphic, chen1998metamorphic, ChenKLPTTZ18, SeguraFSC16} designed to overcome the \textit{oracle problem}, where the ground truth is difficult to determine.
Its core idea is that even if the correct output (i.e., the ``ground truth'') of a single execution is unknown, we can use metamorphic relations between multiple inputs to verify whether their corresponding outputs satisfy the expected invariance.
If this expected relation is violated, it indicates a flaw in the system \cite{chen1998metamorphic}.
In the context of DBMS testing, this idea is materialized in the construction of semantically equivalent query pairs.
The approach rewrites an original query $Q$ into a logically equivalent query $Q'$, which may exhibit a different syntactic structure and execution path.
In this way, the assertion $\text{Result}(Q) \equiv \text{Result}(Q')$ itself constitutes an endogenous (self-contained) test adjudicator that requires no external reference.
Unlike differential testing, meta-transformation testing circumvents the challenge of SQL dialect differences by shifting the focus of verification from inter-system consistency to intra-system logical consistency.

\section{Methodology}

In this section, we present our logical bug detection method, JoinEquiv, based on set-theoretic principles.
Our key insight is that the intersection operation of a set can be expressed in a number of different algebraic equations.
We apply this insight to DBMSs and find that the behavior of \texttt{INNER JOIN} and \texttt{NATURAL JOIN} operations is highly isomorphic to set intersections in the specific but common scenario where the join columns contain no \texttt{NULL} values.

This finding provides a solid theoretical foundation on which we can build our test determiners.
Specifically, we take a simple \emph{original query} $Q_{\text{orig}}$ and use set constants to construct a \emph{transformed query} $Q_{\text{meta}}$.
These two queries are theoretically equivalent in an ideal \texttt{NOT NULL}-constrained environment.
Therefore, any inconsistency between their results clearly indicates a logical bug in the DBMS, rather than semantic ambiguity arising from \texttt{NULL} handling.

\subsection{Framework Overview}

As illustrated in Fig.~\ref{fig:overview}, JoinEquiv consists of six main steps.

Step \circledsmall{1} generates a database state by creating two tables (\texttt{t0}, \texttt{t1}) and populating them with random data, e.g., inserting $\{(1,1), (2,2), (3,3)\}$ into \texttt{t0} and $\{(1), (3), (4)\}$ into \texttt{t1}.

Step \circledsmall{2} creates an \emph{original query} $Q_{\text{orig}}$, ``\texttt{SELECT t0.c0 FROM t0 INNER JOIN t1 ON t0.c0=t1.c0}.''
Since \texttt{INNER JOIN} is a fundamental and well-optimized operator, we expect the DBMS to choose an efficient execution plan using its native query optimizer.
Steps \circledsmall{4a} and \circledsmall{5a} execute $Q_{\text{orig}}$ on the DBMS and fetch its result set.
In our example, assuming the DBMS behaves correctly, the result of $Q_{\text{orig}}$ should be $\text{rs}_{\text{orig}} = \{1, 3\}$, representing the two matching tuples.

Step \circledsmall{3} translates $Q_{\text{orig}}$ into a \emph{transformed query} $Q_{\text{meta}}$.
For instance, by applying the \emph{symmetric join transformation} rule (see Section~\ref{sec:sjt-rule}), we can rewrite $Q_{\text{meta}}$ as \texttt{(SELECT t0.c0 FROM t0 LEFT JOIN t1 ON t0.c0=t1.c0) INTERSECT (SELECT t0.c0 FROM t0 RIGHT JOIN t1 ON t0.c0=t1.c0)}.
Although $Q_{\text{meta}}$ is semantically equivalent to $Q_{\text{orig}}$, its more complex structure (involving two outer joins and a set intersection) effectively disables the optimizer's standard join optimization strategies for a simple inner join.
Consequently, the DBMS is forced to adopt an entirely different and typically more complex execution path, such as first computing the intermediate results of both outer joins and then performing the \texttt{INTERSECT} operation on them.
Steps \circledsmall{4b} and \circledsmall{5b} similarly execute $Q_{\text{meta}}$ on the DBMS and fetch its result set.
In a correctly implemented DBMS, the final result should still be $\text{rs}_{\text{meta}} = \{1, 3\}$.
However, due to an incorrect optimization or an executor defect, some tuples may be erroneously omitted.
In this example, the buggy result $\text{rs}_{\text{bug}} = \{1\}$ misses one tuple.

Finally, Step \circledsmall{6} compares their result sets ($\text{rs}_{\text{orig}}$ vs.~$\text{rs}_{\text{meta}}$ or $\text{rs}_{\text{bug}}$).
For a correct execution, $\text{rs}_{\text{orig}} = \{1, 3\}$ and $\text{rs}_{\text{meta}} = \{1, 3\}$ are identical and the test passes.
In the bug case, $\text{rs}_{\text{bug}} = \{1\}$ indicates that JoinEquiv has successfully detected a logical bug in the query optimizer or executor.

Fundamentally, the effectiveness of our transformation strategy lies in forcing the DBMS to evaluate two fundamentally different execution paths.
While the original \texttt{INNER JOIN} is typically executed via a single highly efficient physical operator (e.g., hash join), our proposed rules construct complex, nested structures involving outer joins and set operations.
This structural complexity often inhibits the optimizer's standard simplification strategies, compelling the DBMS to fully materialize intermediate results and perform combinatorial aggregations.
Any semantic inconsistency that arises between these two logically equivalent but distinct execution paths will manifest as a discrepancy in the outputs.

Next, we describe the three transformation rules in detail.

\subsection{Symmetric Join Transformation}
\label{sec:sjt-rule}

The \textbf{symmetric join transformation (SJT)} rule is grounded in a fundamental principle of relational algebra: 
An \texttt{INNER JOIN} or \texttt{NATURAL JOIN} is logically equivalent to the intersection of its corresponding left and right outer joins under the \texttt{NOT NULL} constraint.
Formally, this rule is derived from the set-theoretic identity $A \cap B$, which represents the minimal canonical form required to reconstruct intersection semantics using asymmetric outer joins.

\begin{figure*}[t]
  \centering
  \includegraphics[width=.9\linewidth]{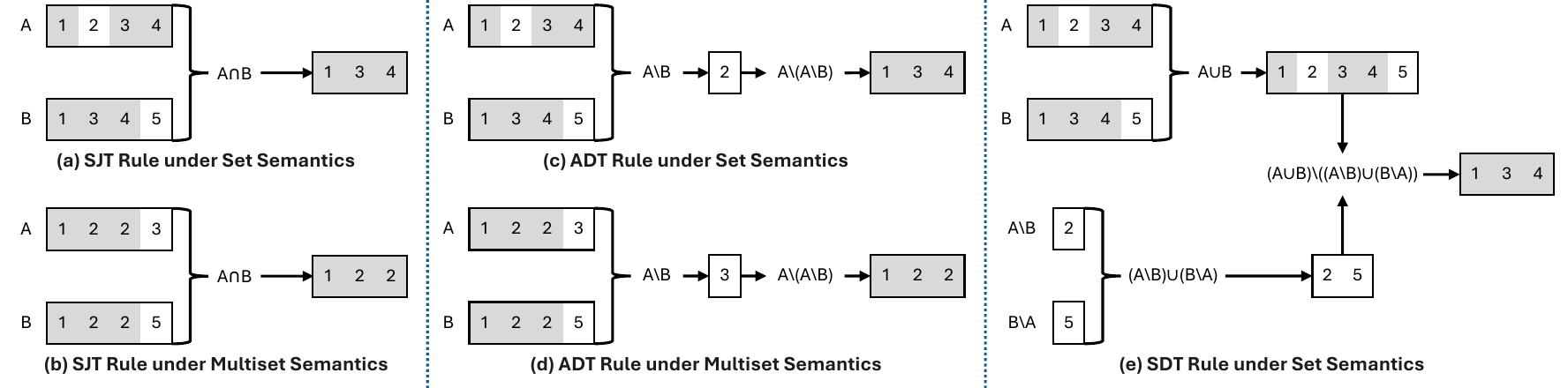}
  \vspace{-1em}
  \caption{Illustrative examples of the SJT, ADT, and SDT transformation rules under set and multiset semantics.\label{fig:data}}
\end{figure*}

To ensure the correctness of transformations, SJT must accurately handle the two core semantics of sets and multisets inherent in SQL.
Using a simple dataset, Figs.~\ref{fig:data}(a) and~\ref{fig:data}(b) demonstrate the difference between set and multiset semantics:
\begin{itemize}
  \item \textbf{Set Semantics:}
  Elements are treated as distinct entities, and duplicate values are automatically removed.
  For instance, with $A = \{1,2,3,4\}$ and $B = \{1,3,4,5\}$, the intersection operation yields $A \cap B = \{1,3,4\}$.
  
  \item \textbf{Multiset Semantics:}
  Element multiplicities are preserved during set operations.  
  Specifically, the \emph{multiplicity} of an element in the result is determined by the \emph{minimum} of its multiplicities in the two input multisets.  
  For instance, with $A = \{1,2,2,3\}$ and $B = \{1,2,2,5\}$, the multiset intersection becomes $A \cap B = \{1,2,2\}$, where the duplicate occurrences of $2$ are retained.
\end{itemize}

Algorithm~\ref{alg:sjt_ast} describes the SJT rule.
It takes as input the abstract syntax tree (AST) of a query and produces a new AST that is structurally different yet semantically equivalent, from which a logically equivalent query can be generated.
First, it verifies whether the input AST contains an \texttt{INNER JOIN} or \texttt{NATURAL JOIN}.
If no such operator exists, the transformation is deemed inapplicable (lines~2--3).
Next, using the original AST as a template, the algorithm constructs two new AST subtrees: one where the join operator is replaced with a \texttt{LEFT JOIN} (denoted as $\text{AST}_L$) and another replaced by a \texttt{RIGHT JOIN} (denoted as $\text{AST}_R$) (lines~5--6).
To preserve duplicate semantics, it inspects whether the query contains the \texttt{DISTINCT} modifier (line~7).
If \texttt{DISTINCT} exists (i.e., using \emph{set semantics}), duplicate tuples are removed and the transformed query employs the \texttt{INTERSECT} operator.
Otherwise (i.e., using \emph{multiset semantics}), duplicates are preserved and \texttt{INTERSECT ALL} is used to maintain tuple multiplicities.
Finally, the algorithm creates a new \texttt{INTERSECT} node as the parent, attaching $\text{AST}_L$ and $\text{AST}_R$ as its left and right children, respectively, thus constructing the transformed $\text{AST}_{\text{out}}$ as output (line~11).

\begin{algorithm}[t]
  \caption{Symmetric Join Transformation (SJT)}
  \label{alg:sjt_ast}
  \footnotesize
  \begin{algorithmic}[1]
    \Statex \textbf{Input:} An input abstract syntax tree $\text{AST}_{\text{in}}$
    \Statex \textbf{Output:} A rewritten abstract syntax tree $\text{AST}_{\text{out}}$ that is semantically equivalent to $\text{AST}_{\text{in}}$
    \State $op \gets \text{GetJoinOperator}(\text{AST}_{\text{in}})$;
    \If{$op \neq \texttt{INNER JOIN}$ \textbf{and} $op \neq \texttt{NATURAL JOIN}$}
      \State \Return $\text{AST}_{\text{in}}$; \Comment{Transformation is not applicable}
    \EndIf
    \State $\text{AST}_L \gets \text{RewriteJoinType}(\text{AST}_{\text{in}}, \texttt{LEFT JOIN})$;
    \State $\text{AST}_R \gets \text{RewriteJoinType}(\text{AST}_{\text{in}}, \texttt{RIGHT JOIN})$;
    \If{$\text{HasDistinctModifier}(\text{AST}_{\text{in}})$}
      \State $intOp \gets \text{CreateNode}(\texttt{INTERSECT})$;
    \Else
      \State $intOp \gets \text{CreateNode}(\texttt{INTERSECT ALL})$;
    \EndIf
    \State $\text{AST}_{\text{out}} \gets \text{Construct}(intOp, \text{AST}_L, \text{AST}_R)$;
    \State \Return $\text{AST}_{\text{out}}$;
  \end{algorithmic}
\end{algorithm}

\subsection{Asymmetric Difference Transformation}
\label{sec:adt-rule}

The \textbf{asymmetric difference transformation (ADT)} rule constructs equivalent queries by leveraging the asymmetric set-difference identity $A \cap B \equiv A \setminus (A \setminus B)$, implemented via \texttt{EXCEPT} (or \texttt{EXCEPT ALL}).
This represents the minimal canonical form to construct an intersection relying solely on the set-difference operator, grounded in the double complement law.
Compared to the SJT rule, ADT generates a nested query structure that poses unique challenges to the DBMS optimizer in handling and simplifying complex \texttt{EXCEPT} clauses.
To ensure transformation correctness, ADT must be validated under set and multiset semantics.
On a simple dataset, Figs.~\ref{fig:data}(c) and~\ref{fig:data}(d) illustrate how ADT behaves under the two semantics:
\begin{itemize}
  \item \textbf{Set Semantics:}  
  Let $A = \{1,2,3,4\}$ and $B = \{1,3,4,5\}$.
  The inner difference, $A \setminus B$, yields $\{2\}$, and applying the outer expression, $A \setminus (A \setminus B)$, restores the elements common to both sets, resulting in $\{1,3,4\}$, which exactly equals the standard intersection $A \cap B$.

  \item \textbf{Multiset Semantics:}  
  Let $A = \{1,2,2,3\}$ and $B = \{1,2,2,5\}$.  
  Under the multiset (bag) semantics, the first difference, $A \setminus B$, produces $\{3\}$, and the subsequent expression, $A \setminus (A \setminus B)$, yields $\{1,2,2\}$, which is consistent with the multiset intersection result.
\end{itemize}

\begin{algorithm}[t]
  \caption{Asymmetric Difference Transformation (ADT)}
  \label{alg:adt}
  \footnotesize
  \begin{algorithmic}[1]
    \Statex \textbf{Input:} An input abstract syntax tree $\text{AST}_{\text{in}}$
    \Statex \textbf{Output:} A rewritten abstract syntax tree $\text{AST}_{\text{out}}$ that is semantically equivalent to $\text{AST}_{\text{in}}$
    \State $op \gets \text{GetJoinOperator}(\text{AST}_{\text{in}})$;
    \If{$op \neq \texttt{INNER JOIN}$ \textbf{and} $op \neq \texttt{NATURAL JOIN}$}
      \State \Return $\text{AST}_{\text{in}}$; \Comment{Transformation is not applicable}
    \EndIf
    \State $\text{AST}_L \gets \text{RewriteJoinType}(\text{AST}_{\text{in}}, \texttt{LEFT JOIN})$;
    \State $\text{AST}_R \gets \text{RewriteJoinType}(\text{AST}_{\text{in}}, \texttt{RIGHT JOIN})$;
    \If{$\text{HasDistinctModifier}(\text{AST}_{\text{in}})$}
      \State $excOp \gets \text{CreateNode}(\texttt{EXCEPT})$;
    \Else
      \State $excOp \gets \text{CreateNode}(\texttt{EXCEPT ALL})$;
    \EndIf
    \State $\text{AST}_{\text{inner\_except}} \gets \text{Construct}(excOp, \text{AST}_L, \text{AST}_R)$;
    \State $\text{AST}_{\text{out}} \gets \text{Construct}(excOp, \text{AST}_L, \text{AST}_{\text{inner\_except}})$;
    \State \Return $\text{AST}_{\text{out}}$;
  \end{algorithmic}
\end{algorithm}

Algorithm~\ref{alg:adt} describes the ADT rule.
Its initialization (lines~1--6) mirrors that of the SJT rule. 
The subsequent steps diverge: It examines whether the original query contains a \texttt{DISTINCT} modifier, selecting either an \texttt{EXCEPT} (for set semantics) or \texttt{EXCEPT ALL} (for multiset semantics) operator node (\texttt{excOp}) to strictly preserve semantic equivalence (lines~7-10).
The core of the ADT transformation lies in a two-stage tree construction.
First, \texttt{excOp} combines $\text{AST}_L$ and $\text{AST}_R$ to construct a subtree representing the inner difference expression $A \setminus B$ 
($\text{AST}_\text{inner\_except}$, line~11).
Then, \texttt{EXCEPT} is applied again with $\text{AST}_L$ as the left child and $\text{AST}_\text{inner\_except}$ as the right child, resulting in the final AST $\text{AST}_\text{out}$ representing the complete expression $A \setminus (A \setminus B)$ (line~12).

\subsection{Symmetric Difference Transformation}
\label{sec:sdt-rule}

The \textbf{symmetric difference transformation (SDT)} rule constructs equivalent queries based on the symmetric set-difference identity $A \cap B \equiv (A \cup B) \setminus ((A \setminus B) \cup (B \setminus A))$.
This constitutes the minimal canonical form to derive the intersection via the \texttt{UNION} operator, grounded in the inclusion-exclusion principle.
Compared to SJT and ADT, SDT produces a more complex set operation structure, challenging the DBMS optimizer with multiple intermediate computations.
To verify transformation correctness, SDT is examined under both set and multiset semantics. 
Using a simple dataset, Fig.~\ref{fig:data}(e) illustrates the behavior of SDT under set semantics:
\begin{itemize}
  \item \textbf{Set Semantics:}  
  Let $A = \{1,2,3,4\}$ and $B = \{1,3,4,5\}$.  
  The symmetric difference $(A \setminus B) \cup (B \setminus A)$ equals $\{2,5\}$, and subtracting this from the full union $A \cup B = \{1,2,3,4,5\}$ yields $\{1,3,4\}$, which is exactly $A \cap B$.  
  This example demonstrates that SDT correctly preserves intersection semantics through symmetric-difference operations under set semantics.

  \item \textbf{Multiset Semantics:}  
  Under multiset semantics, however, the same procedure fails to maintain equivalence. We formally prove such a non-equivalence in Section~\ref{sec:proofs}.
\end{itemize}

\begin{algorithm}[t]
  \caption{Symmetric Difference Transformation (SDT)}
  \label{alg:sdt_corrected}
  \footnotesize
  \begin{algorithmic}[1]
    \Statex \textbf{Input:} An input abstract syntax tree $\text{AST}_{\text{in}}$
    \Statex \textbf{Output:} A rewritten abstract syntax tree $\text{AST}_{\text{out}}$ that is semantically equivalent to $\text{AST}_{\text{in}}$ under set semantics
    \State $op \gets \text{GetJoinOp}(\text{AST}_{\text{in}})$;
    \If {$op \neq \texttt{INNER JOIN}$ \textbf{and} $op \neq \texttt{NATURAL JOIN}$}
      \State \Return $\text{AST}_{\text{in}}$; \Comment{Transformation is not applicable}
    \EndIf
    \State $\text{AST}_L \gets \text{RewriteJoinType}(\text{AST}_{\text{in}}, \texttt{LEFT JOIN})$;
    \State $\text{AST}_R \gets \text{RewriteJoinType}(\text{AST}_{\text{in}}, \texttt{RIGHT JOIN})$;
    \State $N_{un} \gets \text{CreateNode}(\texttt{UNION})$;
    \State $N_{exc} \gets \text{CreateNode}(\texttt{EXCEPT})$;
    \State $N_{un} \gets \text{CreateNode}(\texttt{UNION})$;
    \State $\text{Diff}_L \gets \text{Construct}(N_{exc}, \text{AST}_L, \text{AST}_R)$;
    \State $\text{Diff}_R \gets \text{Construct}(N_{exc}, \text{AST}_R, \text{AST}_L)$;
    \State $\text{Sym}\_\text{Diff} \gets \text{Construct}(N_{un}, \text{Diff}_L, \text{Diff}_R)$;
    \State $\text{Full}\_\text{Union} \gets \text{Construct}(N_{un}, \text{AST}_L, \text{AST}_R)$;
    \State $\text{AST}_{\text{out}} \gets \text{Construct}(N_{exc}, \text{Full}\_\text{Union}, \text{Sym}\_\text{Diff})$;
    \State \Return $\text{AST}_{\text{out}}$;
  \end{algorithmic}
\end{algorithm}

Algorithm~\ref{alg:sdt_corrected} details the SDT transformation.
Its initialization (lines~1--6) mirrors that of SJT and ADT.
Note that, since SDT does not hold for multisets, we consider only its implementation under set semantics.
Then, set operation nodes are created to implement the symmetric difference transformation.
Specifically, two \texttt{EXCEPT} nodes compute the differences $A \setminus B$ and $B \setminus A$ from the left and right join ASTs (lines~10--11). 
These different nodes are then combined under a \texttt{UNION} node to form the symmetric difference $(A \setminus B) \cup (B \setminus A)$ (lines~12--13).
Finally, another \texttt{EXCEPT} node subtracts this symmetric difference from the full union of the left and right join ASTs (line~14), producing the transformed AST that preserves the intersection semantics.

\subsection{Database and Query Generation}

Generating high-quality databases and queries for testing has been extensively studied and is not the primary concern of this work.
JoinEquiv is compatible with any database and query generator capable of producing valid SQL queries containing \texttt{INNER/NATURAL JOIN} operators.
Here, we provide details of our implementation for reproducibility.
The test case generation builds upon the SQLancer framework, which provides a grammar-aware SQL query generation engine.
We extend SQLancer to satisfy an additional precondition that all columns are \texttt{NOT NULL}.
Specifically, Step~\circledsmall{1} of Fig.~\ref{fig:overview} randomly generates tables and rows using \texttt{CREATE TABLE} and \texttt{INSERT} statements under \texttt{NOT NULL} constraints.
Step~\circledsmall{2} then constructs the \emph{original queries}, where JoinEquiv randomly selects a subset of tables from the generated schema to participate in the join.
The order of table references is also randomized.
Across test iterations, the roles of participating tables (e.g., \texttt{t0} and \texttt{t1}) are permuted by a randomized generation engine. Although the transformation templates use a fixed \texttt{LEFT}/\texttt{RIGHT JOIN} structure, symmetric variants are implicitly instantiated through this randomized table binding. For instance, expressions such as \texttt{t1 LEFT JOIN t0} naturally arise, which are algebraically equivalent to \texttt{t0 RIGHT JOIN t1}. As a result, symmetric forms of SJT, ADT, and SDT are repeatedly exercised over long-running testing without explicitly enumerating all symmetric variants.
Such a table permutation, combined with SQLancer's randomized population of \texttt{ON} predicates, \texttt{WHERE} clauses, and \texttt{SELECT} list expressions, ensures that the framework explores a diverse range of symmetric and asymmetric execution plans.

\section{Theoretical Analysis of Transformation Equivalence}
\label{sec:proofs}

The correctness of our fuzzing methodology relies on the semantic equivalence of the query transformations we employ.
This section provides a theoretical analysis of the equivalence of the three transformation rules: SJT, ADT, and SDT.
We rigorously examine their validity under both set semantics (corresponding to \texttt{SELECT DISTINCT}) and multiset semantics (corresponding to standard \texttt{SELECT}).

\subsection{Preliminaries}

To ensure the clarity of our proofs, we operate under the assumption that all columns in the involved tables are defined as \texttt{NOT NULL}.
This allows us to map SQL operations directly to the classical set and multiset theory without the complexities of three-valued logic.
Let $R$ and $S$ be two relations (tables).
We establish the following mappings from SQL operators to their mathematical counterparts:

\para{Set Semantics.}
This model applies to queries using \texttt{SELECT DISTINCT} or standard set operators (\texttt{UNION}, \texttt{EXCEPT}, \texttt{INTERSECT}), which implicitly perform duplicate elimination.
\begin{itemize}
  \item \texttt{R INNER JOIN S} maps to the theta-join $R \bowtie_\theta S$, where $\theta$ is the join condition.
  
  \item \texttt{R LEFT JOIN S ON $\theta$} maps to the left outer theta-join, denoted $R \bowtie_\theta^l S$ := $(R \bowtie_\theta S) \cup R_{\text{unmatched}}$.
  
  \item \texttt{R RIGHT JOIN S ON $\theta$} maps to the right outer theta-join, denoted $R \bowtie_\theta^r S$ := $(R \bowtie_\theta S) \cup S_{\text{unmatched}}$.

  \item \texttt{UNION}, \texttt{EXCEPT}, and \texttt{INTERSECT} (without \texttt{ALL}) map to standard set operators $\cup$, $\setminus$, and $\cap$. Their \texttt{ALL} variants correspond to multiset operations that preserve duplicates.
\end{itemize}

\para{Multiset Semantics.}
Let $A$ and $B$ be two multisets, and let $t$ be a tuple.
The count of $t$ in a multiset $M$ is denoted by $\mathrm{count}(t, M)$.
The multiset semantics for the SQL \texttt{ALL} set operators are defined as follows:
\begin{align}
    \mathrm{count}(t, A \cap B) &= \min(\mathrm{count}(t, A), \mathrm{count}(t, B)), \\
    \mathrm{count}(t, A \cup B) &= \mathrm{count}(t, A) + \mathrm{count}(t, B), \\
    \mathrm{count}(t, A \setminus B) &= \max(0, \mathrm{count}(t, A) - \mathrm{count}(t, B)),
\end{align}
where the operations on the left correspond to \texttt{INTERSECT ALL}, \texttt{UNION ALL}, and \texttt{EXCEPT ALL}, respectively.

\subsection{Proofs of Equivalence under Set Semantics}
Under set semantics, all three transformation rules are provably equivalent to an \texttt{INNER JOIN}.
The equivalence trivially extends to \texttt{NATURAL JOIN}, as it can be regarded as an \texttt{INNER JOIN} with a specific predicate defined on all common attributes.

\begin{proposition}
The SJT rule is equivalent to an \texttt{INNER JOIN} under set semantics, i.e.,
\[
  (R \bowtie_\theta^l S) \cap (R \bowtie_\theta^r S) \equiv R \bowtie_\theta S.
\]
This equivalence also holds for \texttt{NATURAL JOIN}, which is a special case of $\bowtie_\theta$ where $\theta$ equates all common attributes.
\end{proposition}

\begin{proof}
By definition, the left and right outer joins can be written as
\begin{align}
  R \bowtie_\theta^l S &= (R \bowtie_\theta S) \cup R_{\text{unmatched}}, \label{eq:lhs_outerjoin} \\
  R \bowtie_\theta^r S &= (R \bowtie_\theta S) \cup S_{\text{unmatched}}. \label{eq:rhs_outerjoin}
\end{align}
Intersecting these two sets gives
\begin{align}
  (R \bowtie_\theta^l S) \cap (R \bowtie_\theta^r S)
  & = \big( (R \bowtie_\theta S) \cup R_{\text{unmatched}} \big) \notag \\
  & \quad \cap \big( (R \bowtie_\theta S) \cup S_{\text{unmatched}} \big).
\end{align}

Since the sets $R \bowtie_\theta S$, $R_{\text{unmatched}}$, and $S_{\text{unmatched}}$
are mutually disjoint, their intersection eliminates all unmatched tuples, leaving only the common part $(R \bowtie_\theta S)$. Hence,
\begin{align}
  (R \bowtie_\theta^l S) \cap (R \bowtie_\theta^r S) = R \bowtie_\theta S,
\end{align}
which concludes the proof.
\end{proof}

\begin{proposition}
The ADT rule is equivalent to an \texttt{INNER JOIN} under set semantics, i.e.,
\[
  (R \bowtie_\theta^l S) \setminus \big( (R \bowtie_\theta^l S) \setminus (R \bowtie_\theta^r S) \big) \equiv R \bowtie_\theta S.
\]
\end{proposition}

\begin{proof}
By Eqs.~\eqref{eq:lhs_outerjoin} and \eqref{eq:rhs_outerjoin}, the inner difference simplifies as
\begin{align}
  (R \bowtie_\theta^l S) \setminus (R \bowtie_\theta^r S) = R_{\text{unmatched}},
\end{align}
since the matched tuples $(R \bowtie_\theta S)$ are contained in both sets and the unmatched tuples from $S$ do not affect the difference.

Substituting back, the above expression becomes
\begin{align}
  (R \bowtie_\theta^l S) \setminus R_{\text{unmatched}} = R \bowtie_\theta S,
\end{align}
which proves the equivalence under set semantics.
\end{proof}

\begin{proposition}
The SDT rule is equivalent to an \texttt{INNER JOIN} under set semantics, i.e.,
\begin{align}
  & \big( (R \bowtie_\theta^l S) \cup (R \bowtie_\theta^r S) \big) \notag \\
  & \quad \setminus \big( ((R \bowtie_\theta^l S) \setminus (R \bowtie_\theta^r S)) \cup ((R \bowtie_\theta^r S) \setminus (R \bowtie_\theta^l S)) \big) \notag \\
  & \equiv R \bowtie_\theta S.
\end{align}
\end{proposition}

\begin{proof}
Let $A = R \bowtie_\theta^l S$ and $B = R \bowtie_\theta^r S$. 
By the standard set identity, we have
\[
  A \cap B = (A \cup B) \setminus \big( (A \setminus B) \cup (B \setminus A) \big),
\]
the SDT expression computes exactly $A \cap B$.

From Eqs.~\eqref{eq:lhs_outerjoin} and \eqref{eq:rhs_outerjoin}, we have
\begin{align}
  A \setminus B 
  & = (R \bowtie_\theta S) \cup R_{\text{unmatched}} \setminus \big( (R \bowtie_\theta S) \cup S_{\text{unmatched}} \big) \notag \\
  & = R_{\text{unmatched}},\\
  B \setminus A 
  & = (R \bowtie_\theta S) \cup S_{\text{unmatched}} \setminus\ \big( (R \bowtie_\theta S) \cup R_{\text{unmatched}} \big) \notag \\
  & = S_{\text{unmatched}}.
\end{align}
Then, we have
\begin{align}
  (A \setminus B) & \cup (B \setminus A) = R_{\text{unmatched}} \cup S_{\text{unmatched}},\\
  A \cup B &= (R \bowtie_\theta S) \cup R_{\text{unmatched}} \cup S_{\text{unmatched}}.
\end{align}

Subtracting the symmetric difference from the union eliminates all unmatched tuples, 
leaving only matched tuples, i.e.,
\begin{align}
  (A \cup B) \setminus ((A \setminus B) \cup (B \setminus A)) = R \bowtie_\theta S.
\end{align}
Hence, the SDT rule is equivalent to an \texttt{INNER JOIN} under set semantics.
\end{proof}

\subsection{Proofs of Equivalence under Multiset Semantics}

Under multiset semantics, the validity of the transformation rules diverges.

\begin{proposition}
The SJT and ADT rules remain equivalent to an \texttt{INNER JOIN} under multiset semantics.
\end{proposition}
\begin{proof}
For SJT, the multiplicity of a tuple $t$ in the result is
\begin{align}
  \min(\mathrm{count}(t, R \bowtie_\theta^l S), \mathrm{count}(t, R \bowtie_\theta^r S)),
\end{align}
which correctly preserves the multiplicity of tuples from the inner join while excluding unmatched tuples (with count zero in one of the sets).

For ADT, the identity
\begin{equation*}
\begin{split}
    \max\left(0,\, \mathrm{count}(t, A) - \max(0,\, \mathrm{count}(t, A) - \mathrm{count}(t, B))\right) \\
    = \min(\mathrm{count}(t, A), \mathrm{count}(t, B))
\end{split}
\end{equation*}
holds for all non-negative counts.
Therefore, ADT correctly preserves tuple multiplicities under multiset semantics.
\end{proof}

\begin{proposition}
The SDT rule is \textbf{not} equivalent to an \texttt{INNER JOIN} under multiset semantics.
\end{proposition}
\begin{proof}
Given two multisets $A = \{a, a\}$ and $B = \{a, b\}$, the target result for $A \texttt{ INTERSECT ALL } B$ is $\{a\}$.
Applying the SDT transformation, we have
\begin{enumerate}
    \item Full Union ($A \cup B$): $\{a, a, a, b\}$;
    \item Symmetric Difference ($(A \setminus B) \cup (B \setminus A)$): $\{a\} \cup \{b\} = \{a, b\}$;
    \item Final Result ($(A \cup B) \setminus (\text{Sym. Diff.})$): $\{a, a, a, b\} \setminus \{a, b\} = \{a, a\}$.
\end{enumerate}
Since $\{a, a\} \neq \{a\}$ under multiset semantics, the SDT rule is generally not valid for multisets.
\end{proof}

\noindent\textbf{Remark:}
Finally, we discuss in which cases SDT is valid under multiset semantics.
The most direct sufficient condition for SDT's validity is when the input multisets are duplicate-free, effectively degenerating to set semantics (i.e., for all tuples $t$, $\mathrm{count}(t,\cdot)\in\{0,1\}$).
In this case, our proofs based on set algebra apply directly.
An equivalent statement is that SDT can be safely applied if the input relations have been explicitly de-duplicated (e.g., via a \texttt{DISTINCT} modifier) or if the context guarantees that the relations are unique on the projection of their join key columns.

\section{Experimental Evaluation}

In this section, we evaluate JoinEquiv in terms of its ability to discover new vulnerabilities, as well as its efficiency in exploring the state space of the target DBMSs.
Our evaluation aims at answering the following research questions:
\begin{itemize}
\item \textbf{RQ1:} Can JoinEquiv find real logic bugs in widely used and extensively tested DBMSs?
\item \textbf{RQ2:} How diverse are the logic bugs found?
\item \textbf{RQ3:} Can JoinEquiv find logic bugs that are missed by existing approaches?
\end{itemize}

\begin{table*}[t]
  \centering
  \caption{Summary of Logic Bug Reports and Verification Status in Different RDBMSs}
  \vspace{-1em}
  \label{tab:bug_status_professional}
  \footnotesize 
  \begin{tabular}{cccccccl}
    \toprule
    \textbf{RDBMS} & \textbf{Reported} & \textbf{Verified} & \textbf{Fixed} & \textbf{Intended} & \textbf{Component} & \textbf{Severity} & \textbf{Identifier} \\
    \midrule
    MySQL   & 10 & 8 & 1 & 2 & Optimizer (8) & Critical (8) & \makecell[l]{\texttt{Bug\#118544, Bug\#118684, Bug\#118710,} \\ \texttt{Bug\#118857, Bug\#118858, Bug\#118949,} \\ \texttt{Bug\#119032, Bug\#119059}} \\
    \addlinespace
    TiDB    & 13 & 13 & 4 & 0 & \makecell[c]{Planner (11) \\ Execution (2)} & Critical (3) & \makecell[l]{\texttt{\#62380, \#62444, \#62456, \#62459, \#62460,} \\\texttt{\#62644, \#62645, \#62689, \#63596, \#63601,} \\\texttt{\#63635, \#63636, \#63736}} \\
    \addlinespace
    Percona &  3 & 3 & 0 & 0 & Optimizer (3) & Critical (3) & \texttt{PS-10124, PS-10127, DISMYSQL-535} \\
    \addlinespace
    DuckDB & 3 & 3 & 3 & 0 & -- & -- & \texttt{\#20483, \#20486, \#20608} \\
    \cmidrule(r){1-8}
    \textbf{Total} & 29 & 27 & 8 & 2 & -- & Critical (14) & -- \\
    \bottomrule
  \end{tabular}
  \par
  \vspace{0.2em}
  \raggedright 
  \footnotesize
  \textbf{Note:} For DuckDB, component and severity were not provided, so they are marked as ``--''. \textbf{Verified} denotes that the bugs have been officially confirmed and accepted by the developers as valid logic defects.
\end{table*}

\subsection{Evaluation Setup}

\para{Tested DBMSs.}
We tested seven widely used and representative RDBMSs using JoinEquiv, namely MySQL, TiDB, Percona, PostgreSQL, CockroachDB, and SQLite.
These systems are selected as they collectively represent widely deployed single-node RDBMSs, MySQL-compatible variants, distributed SQL systems, and lightweight embedded databases.

\para{Environment.}
All experiments were conducted on a server running 64-bit Ubuntu 20.04.6 LTS.
The machine is equipped with two Intel\textsuperscript{\textregistered} Xeon\textsuperscript{\textregistered} Gold 5218R CPUs (@ 2.10\,GHz), each providing 20 physical cores (40 threads per socket), for a total of 80 hardware threads.

\para{Implementation.}
To ensure the robustness of results and mitigate the impact of randomness, we adopted a \emph{multi-instance parallel testing strategy}.
For each DBMS, we leverage the testing framework's multi-threading capability by launching multiple concurrent testing threads.
Each thread operates with an independent random state and query generation trajectory.
Within each thread, for every \texttt{INNER JOIN} query, JoinEquiv, TLP, and DQP are applied sequentially under the same system configurations and runtime environments.
The testing process ran continuously for 12 hours, effectively simulating repeated experiments across diverse random seeds.

\subsection{Overall Results for Logic Bug Detection}

We test all the target DBMSs with JoinEquiv.
The testing process adheres to a conventional software defect detection pipeline, including fuzzing-based test generation, test case reduction, manual deduplication, cross-DBMS verification, and issue reporting to developers.
As a result, 27 previously unknown logic inconsistencies were detected on MySQL, TiDB, Percona, and DuckDB (see Table~\ref{tab:bug_status_professional}).
Interestingly, we did not uncover many bugs in PostgreSQL, CockroachDB, or SQLite.
The underlying reasons for this ``bug-free'' behavior are analyzed in \S\ref{sec::dis}.

Among the reported bugs, Table~\ref{tab:bug_status_professional} summarizes the results of applying JoinEquiv to different DBMSs.
14 bugs have been verified as \emph{critical}, and 10 additional bugs in TiDB are still pending official severity classification. Additionally, there are two further test cases identified as intended behavior.
Note that different DBMSs adopt different severity levels; for consistency, we use the term \emph{critical} to unify the meanings of \emph{critical} and \emph{serious} in MySQL, \emph{critical} and \emph{major} in TiDB, and \emph{critical} and \emph{urgent} in Percona.
These results demonstrate that JoinEquiv can effectively discover real logic bugs in widely used and extensively tested DBMSs (answering \textbf{RQ1}).

\subsection{Diversity of Detected Logic Bugs}

As shown in Fig.~\ref{fig:dis}, the logic bugs detected by JoinEquiv exhibit remarkable diversity across both transformation rules and join types. Specifically, the three transformation rules demonstrate complementary detection capabilities across different DBMSs, collectively uncovering 27 distinct logic bugs in \texttt{INNER/NATURAL JOIN} operations.
For example, Listing~\ref{lst:MySQL_bug_example} illustrates a bug exposed by SJT, which reveals an inconsistency in NULL semantics handling; Listing~\ref{lst:Percona_bug_example} presents a bug triggered by ADT, exposing a logical vulnerability in predicate rewriting within the optimizer; and Listing~\ref{lst:TiDB_bug_example} shows an executor-level defect in TiDB discovered by SDT, leading to data representation inconsistencies.
Furthermore, the diverse bug distribution across MySQL, TiDB, Percona, and DuckDB indicates that JoinEquiv is independent of system-specific implementation details, as each DBMS exhibits distinct failure modes under different transformation rules.
Together, these findings demonstrate that JoinEquiv can uncover a broad, heterogeneous spectrum of logic bugs originating from multiple layers of the system, including the optimizer, executor, and set-operator implementations, thereby effectively answering \textbf{RQ2}.

\begin{figure}[t]
    \centering
    \includegraphics[width=.98\linewidth]{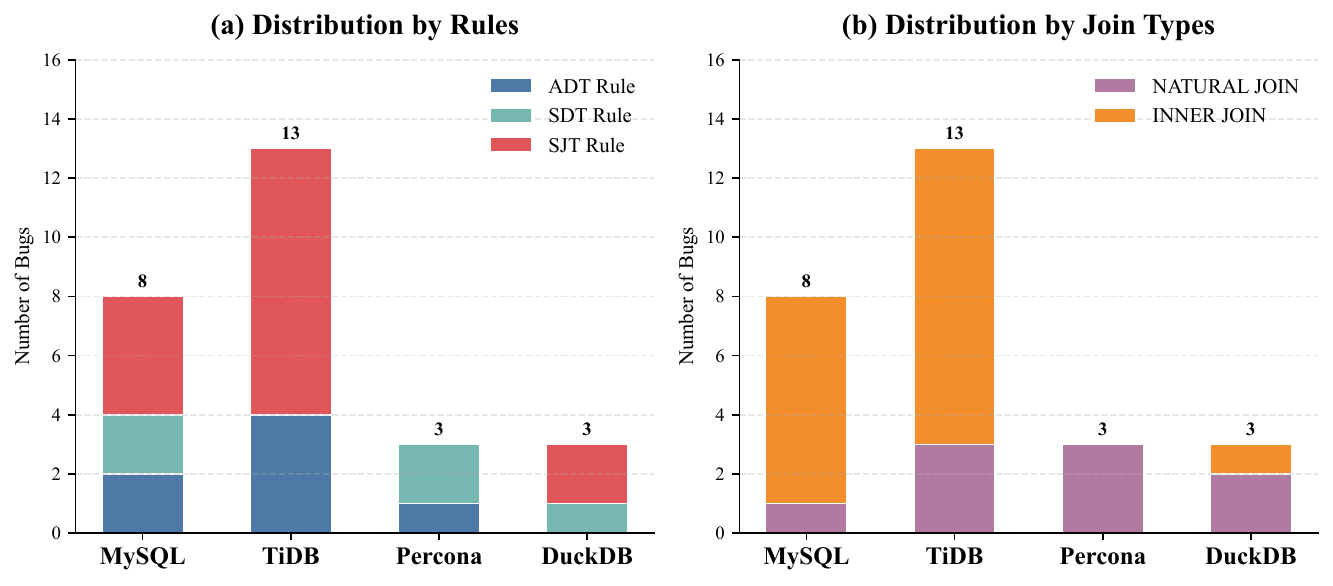}
    \vspace{-1em}
    \caption{Distribution and diversity of logic bugs found by JoinEquiv among different transformation rules and DBMSs.\label{fig:dis}}
\end{figure}

\subsection{Case Analysis for Detected Logic Bugs}

\para{Bug Triggered by SJT Rule.}
As shown in Listing~\ref{lst:MySQL_bug_example}, to investigate the root cause of the bug, we manually analyzed the execution process of the transformed query.
We first executed its two sub-queries independently to verify their correctness.
``\texttt{SELECT DISTINCT t0.c0 FROM t1 NATURAL LEFT JOIN t0}'' returned \texttt{NULL}, correctly filling unmatched \texttt{t1} rows with \texttt{NULL}, 
while ``\texttt{SELECT DISTINCT t0.c0 FROM t1 NATURAL RIGHT JOIN t0}'' returned \texttt{0}, correctly preserving the \texttt{t0} row.

\begin{lstlisting}[
  style=sqlstyle,
  caption={Incorrect 0 and NULL handling in the INTERSECT operator of MySQL.},
  label={lst:MySQL_bug_example}
]
CREATE TABLE t0(c0 BIGINT NOT NULL, c1 DECIMAL NOT NULL);
CREATE TABLE t1(c0 FLOAT NOT NULL);
INSERT INTO t0(c0, c1) VALUES(0, -1);
INSERT INTO t1(c0) VALUES(0.3);
(*@\colorbox{yellow!65}{\color{codegray}Original query}@*)
SELECT DISTINCT t0.c0 FROM t1 NATURAL JOIN t0;
-- empty set (*@\checkicon@*)
(*@\colorbox{yellow!65}{\color{codegray}Transformed query (SJT rule)}@*)
(SELECT DISTINCT t0.c0 FROM t1 NATURAL LEFT JOIN t0)
 INTERSECT
(SELECT DISTINCT t0.c0 FROM t1 NATURAL RIGHT JOIN t0); 
-- {0} (*@\bugicon@*)
\end{lstlisting}

Since both sub-queries return correct results, the error must originate from the final \texttt{INTERSECT} operator.
To confirm this, we analyzed the query plan generated by MySQL for $Q_{\text{meta}}$, as shown in Listing~\ref{lst:intersect-plan}. 
The execution plan shows that the \texttt{INTERSECT} operator receives two single-row inputs (\texttt{NULL} and \texttt{0}). 
However, it produces one row (\texttt{rows=1}), which contradicts the standard semantics of \texttt{INTERSECT}.
According to the SQL standard, \texttt{INTERSECT} should not treat \texttt{NULL} as equal to any non-\texttt{NULL} value during set matching. 
Therefore, the correct result of \{\texttt{NULL}\} \texttt{INTERSECT} \{\texttt{0}\} should be an \texttt{empty set}.
However, MySQL's executor incorrectly treats \texttt{NULL} as equivalent to \texttt{0} during the implementation of the \texttt{INTERSECT} operator, leading to an incorrect result.

\begin{lstlisting}[
  language=SQL,
  caption={Execution plan of the transformed query showing incorrect \texttt{INTERSECT} handling in MySQL.},
  label={lst:intersect-plan}]

-> Table scan on <intersect temporary> (*@\colorbox{red!30}{{(\color{codered}\bfseries rows}=1)}@*)
    -> Intersect materialize with deduplication
        -> (*@{\color{codered}\bfseries Left hash join}@*) (t1.c0 = t0.c0) (rows=1)
            -> ...
        -> (*@{\color{codered}\bfseries Left hash join}@*) (t0.c0 = t1.c0) (rows=1)
            -> ...
\end{lstlisting}

This case reveals a logical bug in MySQL's implementation of set operators: the \texttt{INTERSECT} operator fails to correctly handle the special semantics of \texttt{NULL} under SQL's three-valued logic.
While each individual outer join behaves correctly, the combined set operation introduces a semantic violation, exposing an inconsistency in MySQL's treatment of \texttt{NULL} during set-based evaluation.

\para{Bug Triggered by ADT Rule.}
ADT rule successfully revealed a subtle logical bug deep within the query optimizer. 
In Listing~\ref{lst:Percona_bug_example}, the original query contained a join predicate \texttt{t0.c0 < -0.1}, which is always false; thus, the theoretically correct result should be an empty set.
However, Percona Server incorrectly returned a non-empty result \(\{0, 1\}\).
Through an in-depth analysis of the execution plan for the original query, we identified the root cause: The optimizer performed an unsound predicate rewrite. 
In Listing~\ref{lst:except-plan}, it erroneously relaxes the filter condition \texttt{t0.c0 < -0.1} to \texttt{t0.c0 <= 0}.
This flawed optimization directly led to the non-empty result.

In contrast, the semantically equivalent complex query generated by the ADT rule successfully inhibited the defective optimization, forcing the DBMS to follow a more conservative and correct execution path and correctly returned an empty set.

\vspace{-0.1em}
\begin{lstlisting}[
  style=sqlstyle,
  caption={Incorrect predicate rewrite bug in the optimizer leading to a non-empty result in MySQL.},
  label={lst:Percona_bug_example},
]

CREATE TABLE t0(c0 DECIMAL PRIMARY KEY);
CREATE TABLE t1 LIKE t0;
INSERT INTO t0(c0) VALUES(0.28981613730687783);
INSERT INTO t1(c0) VALUES(0.9690807743979165);
(*@\colorbox{yellow!65}{\color{codegray}Original query}@*)
SELECT t0.c0, t1.c0 FROM t0 INNER JOIN t1 ON t0.c0<-0.1; 
-- {0|1} (*@\bugicon@*)
(*@\colorbox{yellow!65}{\color{codegray}Transformed query(ADT rule)}@*)
(SELECT t0.c0, t1.c0 FROM t0 LEFT JOIN t1 ON t0.c0<-0.1)
 EXCEPT ALL(
    (SELECT t0.c0, t1.c0 FROM t0 LEFT JOIN t1 ON t0.c0<-0.1)
     EXCEPT ALL
    (SELECT t0.c0, t1.c0 FROM t0 RIGHT JOIN t1 ON t0.c0<-0.1));
-- empty set (*@\checkicon@*)
\end{lstlisting}

\begin{lstlisting}[
  language=SQL,
  caption={Execution plan of the original query revealing an unsound predicate rewrite in MySQL.},
  label={lst:except-plan}]

Inner hash join (no condition)  (rows=1)
    -> Covering index scan on t1 using PRIMARY
    -> Hash
        -> Filter:(*@\colorbox{red!30}{(t0.c0 <= 0)}@*)(rows=1)
            -> ..
\end{lstlisting}

\para{Bug Triggered by SDT Rule.}
As shown in Listing~\ref{lst:TiDB_bug_example}, the SDT rule revealed a severe defect of \emph{data representation inconsistency} hidden in the TiDB query executor.
In this test case, the join predicate\footnote{The full join condition: \url{https://github.com/pingcap/tidb/issues/63736}} \emph{$\phi$} is a randomly generated but logically tautological expression (always evaluating to \texttt{TRUE}).
Therefore, the theoretically correct result of all three semantically equivalent queries should be \{\texttt{'7'} $|$ \texttt{'-'}\}.
However, TiDB returned three mutually inconsistent results for the three equivalent queries:
$Q_{\text{orig}}$ returned \{\texttt{0x37} $|$ \texttt{'-'}\}, incorrectly printing the character \texttt{7} as its hexadecimal byte value; $Q_{\text{orig}}$ transformed by the SJT rule returned \{\texttt{'7'} $|$ \texttt{'-'}\}, which happens to be the correct result; $Q_{\text{orig}}$ transformed by the SDT rule returned \{\texttt{0x37000000} $|$ \texttt{'-'}\}, showing a typical case of \emph{data corruption}, where the single-byte representation of \texttt{'7'} was erroneously padded with zero bytes.

This phenomenon clearly exposes a severe implementation inconsistency within TiDB when handling \texttt{CHAR}-typed data across different execution paths. 
By further decomposing the SDT-transformed query, we localized the corruption to the first \texttt{UNION} operation, suggesting that the defect was triggered during the merging of two intermediate result sets from outer joins. This indicates a deep-seated bug in TiDB's \emph{type metadata propagation} or memory layout management mechanisms.

\subsection{Comparison to Existing DBMS Testing Approaches}
\label{subsec:exist}

In order to evaluate JoinEquiv's unique ability to detect logical bugs related to \texttt{INNER JOIN}, we compare it with two representative metamorphic testing approaches: Ternary Logic Partitioning (TLP) and Differential Query Plans (DQP).
TLP shares a conceptual foundation with \textit{JoinEquiv}, as both leverage set-theoretic principles to construct semantically equivalent query variants, whereas DQP evaluates the robustness of query optimizers by executing the same query under different physical plans, a strategy that has proven highly effective in uncovering join-related logical bugs.

\begin{lstlisting}[
  style=sqlstyle,
  caption={A data representation inconsistency under complex set operations in TiDB.},
  label={lst:TiDB_bug_example}
]

CREATE TABLE t0(c0 CHAR NOT NULL);
CREATE TABLE t1 LIKE t0;
INSERT INTO t0(c0) VALUES ('-');
INSERT INTO t1(c0) VALUES ('7');
(*@\colorbox{yellow!65}{\color{codegray}Original query}@*)
(SELECT DISTINCT t1.c0, t0.c0 FROM t1 INNER JOIN t0 ON (*@$\phi$@*));
-- {0x37|'-'} (*@\bugicon@*)
(*@\colorbox{yellow!65}{\color{codegray}Transformed query(ADT rule)}@*)
(SELECT DISTINCT t1.c0, t0.c0 FROM t1 LEFT JOIN t0 ON (*@$\phi$@*)) 
 INTERSECT 
(SELECT DISTINCT t1.c0, t0.c0 FROM t1 RIGHT JOIN t0 ON (*@$\phi$@*));
-- {'7'|'-'} (*@\checkicon@*)
(*@\colorbox{yellow!65}{\color{codegray}Transformed query(SDT rule)}@*)
((SELECT DISTINCT t1.c0 FROM t1 LEFT JOIN t0 ON (*@$\phi$@*))
  UNION
 (SELECT DISTINCT t1.c0 FROM t1 RIGHT JOIN t0 ON (*@$\phi$@*))) 
  EXCEPT(
  ((SELECT DISTINCT t1.c0 FROM t1 LEFT JOIN t0 ON (*@$\phi$@*))
    EXCEPT
   (SELECT DISTINCT t1.c0 FROM t1 RIGHT JOIN t0 ON (*@$\phi$@*)))
    UNION
  ((SELECT DISTINCT t1.c0 FROM t1 RIGHT JOIN t0 ON (*@$\phi$@*))
    EXCEPT
   (SELECT DISTINCT t1.c0 FROM t1 LEFT JOIN t0 ON (*@$\phi$@*))));
-- {0x37000000|'-'} (*@\bugicon@*)
\end{lstlisting}

To ensure a fair and systematic comparison, we conduct two complementary experiments.
First, we perform a \textbf{cross-oracular validation} based on previously confirmed bugs detected by \textit{JoinEquiv}.
For each reported bug, we reconstruct equivalent test cases following TLP and DQP formulations to examine whether these approaches can reproduce the same erroneous behavior.
This validation allows us to assess the orthogonality and uniqueness of \textit{JoinEquiv}'s detection scope.
Second, we carry out a \textbf{unified cross-oracular evaluation}, in which all approaches are executed side by side for 12 hours under identical query generation and database states, to quantitatively compare their bug-finding capabilities. 
Crucially, for every generated original \texttt{INNER JOIN} query that can be correctly executed, we apply all three transformation strategies (TLP, DQP, and JoinEquiv) within the same iteration,  ensuring a fair comparison under an identical execution context.

To evaluate the effectiveness of JoinEquiv, we implemented two representative metamorphic testing approaches as baselines.
Ternary Logic Partitioning (TLP) (specifically, the \textsc{TLP-Where} variant) \cite{tlp} relies on the set-theoretic principle of \emph{UNION}.
As illustrated in Listing~\ref{lst:tlp-example}, TLP partitions an original query into three sub-queries by injecting predicates into the \texttt{WHERE} clause based on SQL's ternary logic (\texttt{TRUE}, \texttt{FALSE}, \texttt{NULL}) and validates that their combined result matches the original.
Differential Query Plans (DQP)~\cite{dqp}, in contrast, targets join optimizer robustness.
It forces the DBMS to generate diverse physical plans for the same query (via hints, see Listing~\ref{lst:tlp-example}) and checks for discrepancies ($R_i \neq R_j$) among the result sets.

\para{Cross-Oracular Validation of Reported Bugs.}
To assess the orthogonality of JoinEquiv's detection scope, we first perform a \emph{cross-oracular validation} on previously reported bugs. 
For each bug exposed by JoinEquiv, we manually rewrite the corresponding query into its TLP- and DQP-equivalent forms and execute them on the same database.

\begin{lstlisting}[
  style=sqlstyle,language=SQL,
  caption={Example of ternary logic partitioning (TLP)} and differential query plans (DQP).,
  label={lst:tlp-example}]

(*@\colorbox{yellow!65}{\color{codegray}Original query}@*)
SELECT * FROM t0 INNER JOIN t1 ON t0.c0 = t1.c0;
(*@\colorbox{yellow!65}{\color{codegray}TLP (Where) transformed query}@*)
SELECT * FROM t0 INNER JOIN t1 ON t0.c0 = t1.c0 WHERE TRUE
UNION ALL
SELECT * FROM t0 INNER JOIN t1 ON t0.c0 = t1.c0 WHERE FALSE 
UNION ALL
SELECT * FROM t0 INNER JOIN t1 ON t0.c0 = t1.c0 WHERE TRUE IS NULL;
(*@\colorbox{yellow!65}{\color{codegray}DQP transformed query}@*)
(*@\color{codegray}{Plan variant 1 (hash join) hint}@*)
SELECT /*+ HASH_JOIN(t0, t1) */ * FROM t0 INNER JOIN t1 ON t0.c0 = t1.c0;
(*@{\color{codegray}{Plan variant 2 (nested-loop join) hint}@*)
SELECT /*+ NESTED_LOOP_JOIN(t0, t1) */ * FROM t0 INNER JOIN t1 ON t0.c0 = t1.c0;
\end{lstlisting}

Formally, given an \texttt{INNER JOIN} query $Q_{\text{orig}}$ whose result differs from one or more of its algebraically equivalent transformed queries, we attempt to construct three predicate-partitioned queries according to TLP's formulation, and generate multiple physically equivalent variants of queries $Q_{\text{DQP}}$ using optimizer hints.
We leverage 32 optimizer hints for MySQL and Percona, and 22 for TiDB (sourced from SQLancer-DQP); DQP does not support DuckDB.  
If the same inconsistency appears, the bug is considered reproducible by TLP or DQP; Otherwise, it is classified as a \emph{Join-specific} bug.
Among the 27 bugs originally detected by JoinEquiv, none were reproduced by either baseline, demonstrating that JoinEquiv possesses a fundamental and irreplaceable advantage over existing approaches in the specific domain of \texttt{INNER JOIN} logic bugs. 

\begin{lstlisting}[
  style=sqlstyle,
  language=SQL,
  caption={TLP transformed query for Listing 3 and Listing 6},
  label={lst:tlp36}]

-- (*@\colorbox{yellow!65}{\color{codegray}TLP Transformed query for Listing 3}@*)
SELECT DISTINCT t0.c0 FROM t1 NATURAL JOIN t0 WHERE TRUE
UNION
SELECT DISTINCT t0.c0 FROM t1 NATURAL JOIN t0 WHERE FALSE
UNION
SELECT DISTINCT t0.c0 FROM t1 NATURAL JOIN t0 WHERE TRUE IS NULL; -- empty set (*@\checkicon@*)
-- (*@\colorbox{yellow!65}{\color{codegray}TLP Transformed query for Listing 6}@*)
SELECT t0.c0, t1.c0 FROM t0 INNER JOIN t1 ON t0.c0<-0.1 WHERE TRUE
UNION ALL
SELECT t0.c0, t1.c0 FROM t0 INNER JOIN t1 ON t0.c0<-0.1 WHERE FALSE
UNION ALL
SELECT t0.c0, t1.c0 FROM t0 INNER JOIN t1 ON t0.c0<-0.1 WHERE TRUE IS NULL; -- {0|1} (*@\bugicon@*)
\end{lstlisting}

During this validation, we observed an interesting phenomenon when applying the TLP strategy.
As shown in Listing~\ref{lst:MySQL_bug_example}, the original \texttt{NATURAL JOIN} query correctly returns an empty set, whereas its SJT-transformed counterpart incorrectly produces \{\texttt{0}\}.
We then rewrote the original query following the TLP's form, partitioning it into its corresponding sub-queries (as shown in Listing~\ref{lst:tlp36}).
Interestingly, all partitions also correctly return empty sets, and their union remains consistent with the original query result.
In contrast, as shown in Listing~\ref{lst:Percona_bug_example}, the original query incorrectly returns \{\texttt{0}$|$\texttt{1}\}.
When rewritten using the TLP approach, the partitioned queries yield the same erroneous result \{\texttt{0} $|$ \texttt{1}\},
whereas the ADT-transformed query correctly produces an empty set.
TLP cannot address the optimizer's semantic bias due to structural transformations during join rewriting.

\para{Unified Cross-Oracular Evaluation.}
To ensure a fair and quantitative comparison between metamorphic testing approaches, we implemented a unified framework that executes JoinEquiv, TLP, and DQP side by side with identical query generation and database states.
All three approaches share the same query generator, database state, random seeds, and execution environment, guaranteeing identical experimental conditions.
For each original query, all three transformations are applied sequentially, and the results are compared across the same DBMS instance.
This setup ensures a fair and controlled environment for large-scale evaluation.

\begin{table}[t]
  \centering
  \caption{Number of bugs reported in a 12h unified run (DQP results exclude DuckDB)}
  \vspace{-1em}
  \label{tab:cross_oracular_results}
  \begin{tabular}{cccc}
    \toprule
    \textbf{DBMS} &  \textbf{TLP} & \textbf{DQP} & \textbf{JoinEquiv} \\
    \midrule
    TiDB    & 0 & 0  & 7  \\
    MySQL   & 0 & 1  & 13 \\
    Percona & 0 & 2  & 3  \\
    DuckDB  & 0 & -- & 13 \\
    \midrule
    \textbf{Total} & 0 & 3 & 36 \\
    \midrule
    \textbf{Increment} & 36$\uparrow$ & 20$\uparrow$ & -- \\
    \bottomrule
  \end{tabular}
\end{table}

Table~\ref{tab:cross_oracular_results} summarizes the number of bugs found in the 12-hour run on MySQL, TiDB, Percona, and DuckDB.
It shows that JoinEquiv found 36 more join-related inconsistencies than TLP and 20 more than DQP (without DuckDB).
TLP and DQP focus on predicate-level or plan-level consistency. 
In contrast, JoinEquiv operates at the algebraic level of join semantics, leveraging simple yet powerful set-theoretic identities.
JoinEquiv systematically verifies the semantic correctness of join rewrites and set-operator interactions, which predicate partitioning or plan variation cannot cover. 
As a result, JoinEquiv can expose deep semantic inconsistencies arising from incorrect join rewriting, NULL propagation, and set operator evaluation.
In summary, JoinEquiv effectively uncovers logic bugs that advanced state-of-the-art approaches fail to detect, thereby answering \textbf{RQ3}.

\para{Code Coverage.}
We measure the code coverage achieved by JoinEquiv on each evaluated DBMS. Over a 12-hour period using 10 threads, it achieves line (statement) coverage of 26.7\%, 25.3\%, and 19.9\% for MySQL, Percona, and TiDB (statement), respectively\footnote{DuckDB's code coverage cannot be directly obtained because it is not straightforward to instrument its build and execution pipeline with standard coverage collection tools.}.
The relatively low code coverage mainly stems from our focus on the query optimization component rather than the full DBMS. 
It is also worth mentioning that prior work has reported no positive correlation between code coverage and testing effectiveness for DBMSs.

\section{Discussion}
\label{sec::dis}

\para{Handling of NULL Semantics.}
JoinEquiv currently restricts all columns to be declared as \texttt{NOT NULL}.
This is a deliberate design choice due to the fundamental semantic mismatch between SQL's three-valued logic (3VL) and set operators.
In the SQL standard, equality comparison predicates involving \texttt{NULL} are evaluated to \texttt{UNKNOWN} and therefore excluded from \texttt{INNER JOIN} results, whereas the set operator treats \texttt{NULL} as equal at the set level and preserves them.
As a result, certain equivalences that hold under classical relational algebra no longer hold in the presence of \texttt{NULL}. 
To avoid introducing spurious inconsistencies caused by semantic divergence, we restrict the current implementation to \texttt{NOT NULL} columns.
Extending JoinEquiv to explicitly reconcile 3VL with set semantics is an interesting direction for future work.

\para{False Positives and Type-Related Bugs.}
Ideally, JoinEquiv relies on proven set-theoretic identities and implies that detected discrepancies represent genuine bugs.
Nevertheless, false positives can arise from the subtle gap between relational algebra semantics and practical DBMS implementations.
For example, we observed representation-level discrepancies such as mixed-type numeric comparisons in MySQL or signed-zero variations (e.g., ``\texttt{-0}'' vs.~``\texttt{0}'') in PostgreSQL, where values are logically equivalent but binary distinct.
Similarly, false positives may stem from the loss of non-functional metadata, such as MySQL's \texttt{ZEROFILL} display attribute being discarded during intermediate set operations.
Consequently, these cases reflect the boundary where DBMS implementations prioritize engineering trade-offs or standard compliance over strict bitwise equivalence.
Beyond potential false positives, a distinct category of detected logic bugs arises from inconsistent type coercion or memory padding that violates semantic invariance and alters the final result set, as demonstrated by the \texttt{NULL} and \texttt{CHAR} handling defects in Listings 3 and 7.

\para{Analysis of Limited Bug Findings in Certain DBMSs.}
As mentioned above, we did not discover optimizer logic bugs in PostgreSQL, CockroachDB, or SQLite.
For PostgreSQL, only representation-level discrepancies (e.g., ``\texttt{-0}'' vs.~``\texttt{0}'') were observed.
This result is consistent with the widely recognized robustness of PostgreSQL and its compatible implementations.
Richard Hipp, the primary author of SQLite, has attributed PostgreSQL's high quality to its ``\textit{very elaborate peer review process}''~\cite{winslett2019richard}.
We believe that this culture of strict standard compliance also extends to compatible systems such as CockroachDB, thereby reducing the likelihood of optimizer-level logic bugs.
For SQLite, JoinEquiv identified several inconsistencies between \texttt{JOIN} and set operations.
The developers clarified these as intended behaviors arising from SQLite's unique type system: while \texttt{JOIN} predicates employ affinity-based comparisons, set operators enforce strict storage-class equality.
Although classified as intentional design choices rather than logic bugs, these findings underscore JoinEquiv's capacity to reveal subtle semantic gaps and implementation-specific variations.
For example, MySQL broadly applies implicit type coercion across operators.

\para{On Rule Completeness and Design Trade-offs.}
Regarding the theoretical scope of \textit{JoinEquiv}, we prioritize operator-level structural completeness over an exhaustive enumeration of set-theoretic identities.
Since identities equivalent to intersection are infinite, making mathematical completeness unattainable, we map intersection to all fundamental SQL set operators.
Each rule is designed as a minimal canonical form, reducing unnecessary structural complexity and making the semantic correspondence between the original and transformed queries easier to analyze, thereby facilitating root-cause localization.
This design principle leads us to prioritize single-rule transformations over complex rule compositions.
Composing multiple rules often results in deeply nested queries that complicate semantic attribution.

\para{Bug Fix Latency.}
Although only 8 out of the 27 confirmed bugs have been fixed so far, the fix rate alone does not directly reflect their severity.
Among the reported bugs, the low fix rate stems mainly from the complexity of fixing logic bugs in join optimizers, which often requires modifying core rewrite rules and extensive regression testing to avoid performance regressions.
Another important factor is the timing of bug disclosure.
As most bugs were reported recently, subsequent releases were developed without awareness of them, making it unlikely that fixes can be incorporated in the short term.

\section{Related Work}

\para{Logic Bug Identification for DBMS.}
Among existing DBMS testing studies, the SQLancer family \cite{sqlancer} is the most representative for discovering logic bugs and is the closest to our work.
Its key contribution lies in providing a framework for systematically uncovering deep errors inside DBMSs.
SQLancer targets logic bugs that lead to incorrect query results in DBMSs.
Other representative approaches include PQS~\cite{pqs}, NoREC~\cite{norec}, TLP~\cite{tlp}, DQP~\cite{dqp}, and CERT~\cite{ba2024cert}.
PQS and NoREC have been discussed in \S\ref{sec:intro}, while TLP and DQP were analyzed in \S\ref{subsec:exist}.
CERT~\cite{ba2024cert} aims to identify performance issues caused by unexpected cardinality estimations, that is, the cases where the estimated number of result tuples significantly deviates from the actual number.
In contrast, our approach specifically targets the detection of logic bugs in join optimizations, an area that none of the existing SQLancer-based techniques are designed to explore.

\para{Generator-based Testing for DBMS.}
A variety of tools have been developed for generating database data~\cite{binnig2007qagen, bruno2005flexible, gray1994quickly, houkjaer2006simple, khalek2008query} and SQL queries~\cite{bati2007genetic, bruno2006generating, jung2019apollo, mishra2008generating, poess2004generating, seltenreich2025sqlsmith, vartak2010qrelx} to automatically construct test cases for DBMS evaluation.
However, relatively little attention has been paid to the design of test oracles, which are essential for determining the correctness of query results.
Generation-based testing approaches~\cite{gu2012testing, abdul2010automated, lo2010framework, mishra2008generating, yan2018snowtrail, zhong2020squirrel} have been widely adopted for DBMS validation, particularly for bug detection and benchmarking.
Among them, SQLSmith~\cite{seltenreich2025sqlsmith} is one of the most popular generation-based DBMS testers.
It constructs syntactically correct and complex SQL queries from scratch based on the existing database schema, primarily aiming to detect database crashes and internal errors. 
Recently, TQS~\cite{tang2023detecting} introduced data-guided schema generation to derive ground-truth results from wide tables and knowledge-guided query exploration to efficiently traverse the join search space.
It then detects bugs by comparing the ground truth with the results of query plans enforced by optimizer hints.
Unlike TQS, which relies on specific data synthesis and exploration strategies, \textsc{JoinEquiv} proposes a rigorous metamorphic oracle based on set-theoretic algebra that detects bugs through semantic inconsistencies rather than data lineage.

\para{Mutation-based DBMS Fuzzing.}
Mutation-based fuzzers have long been effective in software testing and have discovered numerous bugs~\cite{aschermann2019redqueen, chen2019enfuzz, liang2019deepfuzzer, liang2022pata, wang2021riff, wang2021industry, wu2022unicorn, yun2018qsym, afl2025}.
Early fuzzers such as AFL~\cite{afl2025} can be adapted to lightweight DBMS components, but random mutations often fail to produce syntactically valid SQL inputs.
Subsequent studies introduced taint analysis~\cite{aschermann2019redqueen, liang2022pata} and symbolic execution~\cite{liang2019deepfuzzer, yun2018qsym} to guide mutation, yet generating test cases that satisfy both syntax and semantics remains challenging.
To address these limitations, recent advances have focused on structure-aware and sequence-oriented strategies.
SQUIRREL~\cite{zhong2020squirrel} proposed a syntax-preserving mutation strategy on an intermediate representation and a semantics-guided instantiation mechanism to satisfy data dependencies, ensuring a high validity of generated queries.
Building on this, LEGO~\cite{liang2023sequence} identifies that the order of SQL statements is crucial for triggering deep DBMS states.
It utilizes type-affinity analysis to synthesize meaningful SQL sequences rather than isolated queries.
Other works also contribute to specific scenarios, such as RATEL~\cite{wang2021industry} for enterprise DBMS feedback, UNICORN~\cite{wu2022unicorn} for time-series databases, and GRIFFIN~\cite{fu2022griffin} for grammar-free mutation.
However, mutation-based fuzzers mainly use code coverage feedback (e.g., discovering new execution paths) to trigger crashes or assertion failures.
They are not well-suited for detecting logic bugs, where the DBMS returns incorrect results without crashing. 
In contrast, our work focuses on the semantic correctness of query results.
By systematically applying algebraic equivalence rules, we generate test oracles that can identify logical inconsistencies in join optimizations that coverage-guided fuzzers would likely miss.

\section{Conclusion}

Building on the basic principles of set theory, this paper introduced JoinEquiv, a novel framework for detecting logic bugs in DBMSs through equivalence-based query transformation.
By viewing the join operator as a generalized intersection under set and multiset semantics, we systematically derived three provably equivalent rewriting rules (SJT, ADT, and SDT) and applied them to uncover inconsistencies overlooked by existing state-of-the-art testing approaches such as TLP and DQP.
Extensive experiments on four DBMSs detected 27 previously unknown bugs and 14 verified as \emph{critical} bugs, demonstrating that even simple algebraic equivalences can serve as powerful metamorphic relations for testing complex database systems.
Future work will explore the application of equivalence-based testing beyond join operations.

\section*{Acknowledgements}

We thank anonymous reviewers for their constructive comments and suggestions, which helped improve this paper.
This work was supported in part by Ant Digital Technologies, Ant Group Research Fund.
This work was also supported in part by the Shanghai Sailing Program No.~23YF1410600.

\section*{AI-Generated Content Acknowledgment}

We claim that large language model (LLM) tools (GPT-4o and Grammarly) were used for text editing, e.g., spelling auto-correction, grammar checks, and sentence rephrasing, in the writing of this paper.
We ensure that no AI-generated content was used without human verification or modification.
We guarantee that no parts of the scientific content---including research design, data analysis, result generation, and code implementation---were generated or affected by AI systems.

\bibliographystyle{IEEEtran}
\bibliography{refs}

@article{codd1970relational,
  title        = {A relational model of data for large shared data banks},
  author       = {Codd, Edgar F},
  year         = {1970},
  journal      = {Commun. ACM},
  volume       = {13},
  number       = {6},
  pages        = {377--387}
}

@inproceedings{selinger1979access,
  title        = {Access path selection in a relational database management system},
  author       = {Selinger, P Griffiths and Astrahan, Morton M and Chamberlin, Donald D and Lorie, Raymond A and Price, Thomas G},
  year         = {1979},
  booktitle    = {Proceedings of the 1979 ACM SIGMOD International Conference on Management of Data},
  pages        = {23--34}
}

@inproceedings{graefe1993volcano,
  title        = {The Volcano Optimizer Generator: Extensibility and Efficient Search},
  author       = {Goetz Graefe and William J. McKenna},
  year         = {1993},
  booktitle    = {Proceedings of the Ninth International Conference on Data Engineering},
  pages        = {209--218}
}

@inproceedings{gray1994quickly,
  title        = {Quickly generating billion-record synthetic databases},
  author       = {Gray, Jim and Sundaresan, Prakash and Englert, Susanne and Baclawski, Ken and Weinberger, Peter J},
  year         = {1994},
  booktitle    = {Proceedings of the 1994 ACM SIGMOD International Conference on Management of Data},
  pages        = {243--252}
}

@inproceedings{slutz1998massive,
  title        = {Massive stochastic testing of {SQL}},
  author       = {Slutz, Donald R},
  year         = {1998},
  booktitle    = {Proceedings of 24rd International Conference on Very Large Data Bases},
  pages        = {618--622}
}

@inproceedings{Chaudhuri1998Overview,
  title        = {An Overview of Query Optimization in Relational Systems},
  author       = {Surajit Chaudhuri},
  year         = {1998},
  booktitle    = {Proceedings of the Seventeenth ACM SIGACT-SIGMOD-SIGART Symposium on Principles of Database Systems},
  pages        = {34--43}
}

@inproceedings{poess2004generating,
  title        = {Generating thousand benchmark queries in seconds},
  author       = {Poess, Meikel and Stephens Jr, John M},
  year         = {2004},
  booktitle    = {Proceedings of the 30th International Conference on Very Large Data Bases},
  pages        = {1045--1053}
}

@inproceedings{bruno2005flexible,
  title        = {Flexible database generators},
  author       = {Bruno, Nicolas and Chaudhuri, Surajit},
  year         = {2005},
  booktitle    = {Proceedings of the 31st International Conference on Very Large Data Bases},
  pages        = {1097--1107}
}

@article{bruno2006generating,
  title        = {Generating queries with cardinality constraints for {DBMS} testing},
  author       = {Bruno, Nicolas and Chaudhuri, Surajit and Thomas, Dilys},
  year         = {2006},
  journal      = {IEEE Trans. Knowl. Data Eng.},
  volume       = {18},
  number       = {12},
  pages        = {1721--1725}
}

@article{howden2006theoretical,
  title        = {Theoretical and empirical studies of program testing},
  author       = {Howden, William E},
  year         = {2006},
  journal      = {IEEE Trans. Softw. Eng.},
  number       = {4},
  pages        = {293--298}
}

@inproceedings{houkjaer2006simple,
  title        = {Simple and realistic data generation},
  author       = {Houkj{\ae}r, Kenneth and Torp, Kristian and Wind, Rico},
  year         = {2006},
  booktitle    = {Proceedings of the 32nd International Conference on Very Large Data Bases},
  pages        = {1243--1246}
}

@inproceedings{bati2007genetic,
  title        = {A genetic approach for random testing of database systems},
  author       = {Bati, Hardik and Giakoumakis, Leo and Herbert, Steve and Surna, Aleksandras},
  year         = {2007},
  booktitle    = {Proceedings of the 33rd International Conference on Very Large Data Bases},
  pages        = {1243--1251}
}

@inproceedings{binnig2007qagen,
  title        = {{QAGen}: generating query-aware test databases},
  author       = {Binnig, Carsten and Kossmann, Donald and Lo, Eric and {\"O}zsu, M Tamer},
  year         = {2007},
  booktitle    = {Proceedings of the 2007 ACM SIGMOD International Conference on Management of Data},
  pages        = {341--352}
}

@inproceedings{khalek2008query,
  title        = {Query-aware test generation using a relational constraint solver},
  author       = {Khalek, Shadi Abdul and Elkarablieh, Bassem and Laleye, Yai O and Khurshid, Sarfraz},
  year         = {2008},
  booktitle    = {2008 23rd IEEE/ACM International Conference on Automated Software Engineering},
  pages        = {238--247}
}

@inproceedings{mishra2008generating,
  title        = {Generating targeted queries for database testing},
  author       = {Mishra, Chaitanya and Koudas, Nick and Zuzarte, Calisto},
  year         = {2008},
  booktitle    = {Proceedings of the 2008 ACM SIGMOD International Conference on Management of Data},
  pages        = {499--510}
}

@article{lo2010framework,
  title        = {A framework for testing {DBMS} features},
  author       = {Lo, Eric and Binnig, Carsten and Kossmann, Donald and Tamer {\"O}zsu, M and Hon, Wing-Kai},
  year         = {2010},
  journal      = {VLDB J.},
  volume       = {19},
  number       = {2},
  pages        = {203--230}
}

@inproceedings{abdul2010automated,
  title        = {Automated {SQL} query generation for systematic testing of database engines},
  author       = {Abdul Khalek, Shadi and Khurshid, Sarfraz},
  year         = {2010},
  booktitle    = {Proceedings of the 25th IEEE/ACM International Conference on Automated Software Engineering},
  pages        = {329--332}
}

@inproceedings{vartak2010qrelx,
  title        = {{QRelX}: generating meaningful queries that provide cardinality assurance},
  author       = {Vartak, Manasi and Raghavan, Venkatesh and Rundensteiner, Elke A},
  year         = {2010},
  booktitle    = {Proceedings of the 2010 ACM SIGMOD International Conference on Management of data},
  pages        = {1215--1218}
}

@inproceedings{gu2012testing,
  title        = {Testing the accuracy of query optimizers},
  author       = {Gu, Zhongxian and Soliman, Mohamed A and Waas, Florian M},
  year         = {2012},
  booktitle    = {Proceedings of the Fifth International Workshop on Testing Database Systems},
  pages        = {1--6}
}

@article{stonebraker2013intel,
  title        = {Intel ``big data'' science and technology center vision and execution plan},
  author       = {Stonebraker, Michael and Madden, Sam and Dubey, Pradeep},
  year         = {2013},
  journal      = {SIGMOD Rec.},
  volume       = {42},
  number       = {1},
  pages        = {44--49}
}

@article{SeguraFSC16,
  title        = {A Survey on Metamorphic Testing},
  author       = {Sergio Segura and Gordon Fraser and Ana Bel{\'{e}}n S{\'{a}}nchez and Antonio Ruiz Cort{\'{e}}s},
  year         = {2016},
  journal      = {IEEE Trans. Softw. Eng.},
  volume       = {42},
  number       = {9},
  pages        = {805--824}
}

@article{ding2018plan,
  title        = {Plan stitch: Harnessing the best of many plans},
  author       = {Ding, Bailu and Das, Sudipto and Wu, Wentao and Chaudhuri, Surajit and Narasayya, Vivek},
  year         = {2018},
  journal      = {Proc. VLDB Endow.},
  volume       = {11},
  number       = {10},
  pages        = {1123--1136}
}

@article{ChenKLPTTZ18,
  title        = {Metamorphic Testing: A Review of Challenges and Opportunities},
  author       = {Tsong Yueh Chen and Fei-Ching Kuo and Huai Liu and Pak-Lok Poon and Dave Towey and T. H. Tse and Zhi Quan Zhou},
  year         = {2018},
  journal      = {ACM Comput. Surv.},
  volume       = {51},
  number       = {1},
  pages        = {4:1--4:27}
}

@article{leis2018query,
  title        = {Query optimization through the looking glass, and what we found running the Join Order Benchmark},
  author       = {Viktor Leis and Bernhard Radke and Andrey Gubichev and Atanas Mirchev and Peter A. Boncz and Alfons Kemper and Thomas Neumann},
  year         = {2018},
  journal      = {VLDB J.},
  volume       = {27},
  number       = {5},
  pages        = {643--668}
}

@inproceedings{neumann2018adaptive,
  title        = {Adaptive Optimization of Very Large Join Queries},
  author       = {Thomas Neumann and Bernhard Radke},
  year         = {2018},
  booktitle    = {Proceedings of the 2018 International Conference on Management of Data},
  pages        = {677--692}
}

@inproceedings{yan2018snowtrail,
  title        = {Snowtrail: Testing with production queries on a cloud database},
  author       = {Yan, Jiaqi and Jin, Qiuye and Jain, Shrainik and Viglas, Stratis D and Lee, Allison},
  year         = {2018},
  booktitle    = {Proceedings of the Workshop on Testing Database Systems},
  pages        = {1--6}
}

@inproceedings{yun2018qsym,
  title        = {{QSYM}: A practical concolic execution engine tailored for hybrid fuzzing},
  author       = {Yun, Insu and Lee, Sangho and Xu, Meng and Jang, Yeongjin and Kim, Taesoo},
  year         = {2018},
  booktitle    = {27th USENIX Security Symposium (USENIX Security 18)},
  pages        = {745--761}
}

@article{jung2019apollo,
  title        = {Apollo: Automatic detection and diagnosis of performance regressions in database systems},
  author       = {Jung, Jinho and Hu, Hong and Arulraj, Joy and Kim, Taesoo and Kang, Woonhak},
  year         = {2019},
  journal      = {Proc. VLDB Endow.},
  volume       = {13},
  number       = {1},
  pages        = {57--70}
}

@article{liang2019deepfuzzer,
  title        = {{DeepFuzzer}: Accelerated deep greybox fuzzing},
  author       = {Liang, Jie and Jiang, Yu and Wang, Mingzhe and Jiao, Xun and Chen, Yuanliang and Song, Houbing and Choo, Kim-Kwang Raymond},
  year         = {2019},
  journal      = {IEEE Trans. Dependable Secur. Comput.},
  volume       = {18},
  number       = {6},
  pages        = {2675--2688}
}

@article{ryan2019neo,
  title        = {Neo: A Learned Query Optimizer},
  author       = {Ryan Marcus and Parimarjan Negi and Hongzi Mao and Chi Zhang and Mohammad Alizadeh and Tim Kraska and Olga Papaemmanouil and Nesime Tatbul},
  year         = {2019},
  journal      = {Proc. VLDB Endow.},
  volume       = {12},
  number       = {11},
  pages        = {1705--1718}
}

@inproceedings{aschermann2019redqueen,
  title        = {{REDQUEEN}: Fuzzing with Input-to-State Correspondence},
  author       = {Aschermann, Cornelius and Schumilo, Sergej and Blazytko, Tim and Gawlik, Robert and Holz, Thorsten},
  year         = {2019},
  booktitle    = {26th Annual Network and Distributed System Security Symposium (NDSS)},
  pages        = {1--15}
}

@inproceedings{chen2019enfuzz,
  title        = {{EnFuzz}: Ensemble fuzzing with seed synchronization among diverse fuzzers},
  author       = {Chen, Yuanliang and Jiang, Yu and Ma, Fuchen and Liang, Jie and Wang, Mingzhe and Zhou, Chijin and Jiao, Xun and Su, Zhuo},
  year         = {2019},
  booktitle    = {28th USENIX Security Symposium (USENIX Security 19)},
  pages        = {1967--1983}
}

@article{tlp,
  title        = {Finding bugs in database systems via query partitioning},
  author       = {Manuel Rigger and Zhendong Su},
  year         = {2020},
  journal      = {Proc. ACM Program. Lang.},
  volume       = {4},
  number       = {OOPSLA},
  pages        = {211:1--211:30}
}

@article{chen1998metamorphic,
  title        = {Metamorphic Testing: {A} New Approach for Generating Next Test Cases},
  author       = {Tsong Yueh Chen and Shing-Chi Cheung and Siu-Ming Yiu},
  year         = {2020},
  journal      = {arXiv:2002.12543}
}

@inproceedings{norec,
  title        = {Detecting optimization bugs in database engines via non-optimizing reference engine construction},
  author       = {Manuel Rigger and Zhendong Su},
  year         = {2020},
  booktitle    = {Proceedings of the 28th ACM Joint European Software Engineering Conference and Symposium on the Foundations of Software Engineering},
  pages        = {1140--1152}
}

@inproceedings{pqs,
  title        = {Testing database engines via pivoted query synthesis},
  author       = {Rigger, Manuel and Su, Zhendong},
  year         = {2020},
  booktitle    = {14th USENIX Symposium on Operating Systems Design and Implementation (OSDI 20)},
  pages        = {667--682}
}

@inproceedings{zhong2020squirrel,
  title        = {{SQUIRREL}: Testing database management systems with language validity and coverage feedback},
  author       = {Zhong, Rui and Chen, Yongheng and Hu, Hong and Zhang, Hangfan and Lee, Wenke and Wu, Dinghao},
  year         = {2020},
  booktitle    = {Proceedings of the 2020 ACM SIGSAC Conference on Computer and Communications Security},
  pages        = {955--970}
}

@inproceedings{wang2021riff,
  title        = {{RIFF}: Reduced instruction footprint for {Coverage-Guided} fuzzing},
  author       = {Wang, Mingzhe and Liang, Jie and Zhou, Chijin and Jiang, Yu and Wang, Rui and Sun, Chengnian and Sun, Jiaguang},
  year         = {2021},
  booktitle    = {2021 USENIX Annual Technical Conference (USENIX ATC 21)},
  pages        = {147--159}
}

@inproceedings{wang2021industry,
  title        = {Industry practice of coverage-guided enterprise-level {DBMS} fuzzing},
  author       = {Wang, Mingzhe and Wu, Zhiyong and Xu, Xinyi and Liang, Jie and Zhou, Chijin and Zhang, Huafeng and Jiang, Yu},
  year         = {2021},
  booktitle    = {2021 IEEE/ACM 43rd International Conference on Software Engineering: Software Engineering in Practice (ICSE-SEIP)},
  pages        = {328--337}
}

@inproceedings{fu2022griffin,
  title        = {Griffin : Grammar-Free {DBMS} Fuzzing},
  author       = {Fu, Jingzhou and Liang, Jie and Wu, Zhiyong and Wang, Mingzhe and Jiang, Yu},
  year         = {2022},
  booktitle    = {Proceedings of the 37th IEEE/ACM International Conference on Automated Software Engineering},
  pages        = {49:1--49:12}
}

@inproceedings{liang2022pata,
  title        = {{PATA}: Fuzzing with Path Aware Taint Analysis},
  author       = {Liang, Jie and Wang, Mingzhe and Zhou, Chijin and Wu, Zhiyong and Jiang, Yu and Liu, Jianzhong and Liu, Zhe and Sun, Jiaguang},
  year         = {2022},
  booktitle    = {2022 IEEE Symposium on Security and Privacy (SP)},
  pages        = {1--17}
}

@inproceedings{wu2022unicorn,
  title        = {Unicorn: detect runtime errors in time-series databases with hybrid input synthesis},
  author       = {Wu, Zhiyong and Liang, Jie and Wang, Mingzhe and Zhou, Chijin and Jiang, Yu},
  year         = {2022},
  booktitle    = {Proceedings of the 31st ACM SIGSOFT International Symposium on Software Testing and Analysis},
  pages        = {251--262}
}

@article{zhang2023simple,
  title        = {Simple adaptive query processing vs. learned query optimizers: Observations and analysis},
  author       = {Zhang, Yunjia and Chronis, Yannis and Patel, Jignesh M and Rekatsinas, Theodoros},
  year         = {2023},
  journal      = {Proc. VLDB Endow.},
  volume       = {16},
  number       = {11},
  pages        = {2962--2975}
}

@article{dqp,
  title        = {Keep It Simple: Testing Databases via Differential Query Plans},
  author       = {Jinsheng Ba and Manuel Rigger},
  year         = {2024},
  journal      = {Proc. ACM Manag. Data},
  volume       = {2},
  number       = {3},
  pages        = {188:1--188:26}
}

@article{xin24spatial,
  title        = {Spatial Query Optimization With Learning},
  author       = {Xin Zhang and Ahmed Eldawy},
  year         = {2024},
  journal      = {Proc. VLDB Endow.},
  volume       = {17},
  number       = {12},
  pages        = {4245--4248}
}

@inproceedings{thanos,
  title        = {{THANOS}: {DBMS} Bug Detection via Storage Engine Rotation Based Differential Testing},
  author       = {Fu, Ying and Wu, Zhiyong and Zhang, Yuanliang and Liang, Jie and Fu, Jingzhou and Jiang, Yu and Li, Shanshan and Liao, Xiangke},
  year         = {2024},
  booktitle    = {2025 IEEE/ACM 47th International Conference on Software Engineering (ICSE)},
  pages        = {1--12}
}

@inproceedings{ba2024cert,
  title        = {{CERT}: Finding Performance Issues in Database Systems Through the Lens of Cardinality Estimation},
  author       = {Jinsheng Ba and Manuel Rigger},
  year         = {2024},
  booktitle    = {Proceedings of the 46th IEEE/ACM International Conference on Software Engineering},
  pages        = {133:1--133:13}
}

@misc{percona,
  title        = {Percona Server for {MySQL}},
  author       = {{Percona LLC}},
  year         = {2025},
  note         = {Accessed: 2025-10-27},
  howpublished = {\url{https://www.percona.com/software/mysql-database/percona-server}}
}

@misc{postgresql,
  title        = {{PostgreSQL}: The World's Most Advanced Open Source Relational Database},
  author       = {{PostgreSQL Global Development Group}},
  year         = {2025},
  note         = {Accessed: 2025-10-27},
  howpublished = {\url{https://www.postgresql.org}}
}

@article{song2025detecting,
  title        = {Detecting Schema-Related Logic Bugs in Relational {DBMS}s via Equivalent Database Construction},
  author       = {Jiansen Song and Wensheng Dou and Yingying Zheng and Yu Gao and Ziyu Cui and Wei Wang and Jun Wei},
  year         = {2025},
  journal      = {Proc. VLDB Endow.},
  volume       = {18},
  number       = {7},
  pages        = {2281--2294}
}

@misc{sqlancer,
  title        = {{SQLancer}},
  author       = {Manuel Rigger},
  year         = {2025},
  note         = {Accessed: 2025-10-27},
  howpublished = {\url{https://github.com/sqlancer/sqlancer}}
}

@misc{tidb,
  title        = {{TiDB}},
  author       = {{PingCAP}},
  year         = {2025},
  note         = {Accessed: 2025-10-27},
  howpublished = {\url{https://pingcap.com/tidb}}
}

@misc{seltenreich2025sqlsmith,
  title        = {{SQLSmith}},
  author       = {Seltenreich, Andreas},
  year         = {2025},
  note         = {Accessed: 2025-10-27},
  howpublished = {\url{https://github.com/anse1/sqlsmith}}
}

@misc{wikipedia_test_oracle_2025,
  title        = {Test oracle},
  author       = {Wikipedia},
  year         = {2025},
  note         = {Accessed: 2025-10-27},
  howpublished = {\url{https://en.wikipedia.org/wiki/Test_oracle}}
}

@misc{afl2025,
  title        = {American Fuzzy Lop},
  author       = {Zalewski, Micha\l{}},
  year         = {2025},
  note         = {Accessed: 2025-10-27},
  howpublished = {\url{https://github.com/google/AFL}}
}

@article{chamberlin1976relational,
  title={Relational data-base management systems},
  author={Chamberlin, Donald D},
  journal={ACM Comput. Surv.},
  volume={8},
  number={1},
  pages={43--66},
  year={1976}
}

@article{stonebraker2024goes,
  title={What Goes Around Comes Around... And Around...},
  author={Stonebraker, Michael and Pavlo, Andrew},
  journal={SIGMOD Rec.},
  volume={53},
  number={2},
  pages={21--37},
  year={2024}
}

@article{hai2025quantum,
  title={Quantum Data Management in the {NISQ} Era},
  author={Hai, Rihan and Hung, Shih-Han and Coopmans, Tim and Littau, Tim and Geerts, Floris},
  journal={Proc. VLDB Endow.},
  volume={18},
  number={6},
  pages={1720--1729},
  year={2025}
}

@article{zhao2025debunking,
  author       = {Junyi Zhao and
                  Kai Su and
                  Yifei Yang and
                  Xiangyao Yu and
                  Paraschos Koutris and
                  Huanchen Zhang},
  title        = {Debunking the Myth of Join Ordering: Toward Robust {SQL} Analytics},
  journal      = {Proc. ACM Manag. Data},
  volume       = {3},
  number       = {3},
  pages        = {146:1--146:28},
  year         = {2025}
}

@article{birler2024robust,
  title={Robust join processing with diamond hardened joins},
  author={Birler, Altan and Kemper, Alfons and Neumann, Thomas},
  journal={Proc. VLDB Endow.},
  volume={17},
  number={11},
  pages={3215--3228},
  year={2024}
}

@article{hu2025output,
  author       = {Xiao Hu},
  title        = {Output-Optimal Algorithms for Join-Aggregate Queries},
  journal      = {Proc. ACM Manag. Data},
  volume       = {3},
  number       = {2},
  pages        = {104:1--104:27},
  year         = {2025}
}

@inproceedings{kalumin2025optimizing,
  title={Optimizing queries with many-to-many joins},
  author={Kalumin, Hasara and Deshpande, Amol},
  booktitle={2025 IEEE 41st International Conference on Data Engineering (ICDE)},
  pages={3668--3681},
  year={2025}
}

@article{tang2023detecting,
  title={Detecting Logic Bugs of Join Optimizations in {DBMS}},
  author={Tang, Xiu and Wu, Sai and Zhang, Dongxiang and Li, Feifei and Chen, Gang},
  journal={Proc. ACM Manag. Data},
  volume={1},
  number={1},
  pages={55:1--55:26},
  year={2023}
}

@inproceedings{liang2023sequence,
  title={Sequence-oriented {DBMS} fuzzing},
  author={Liang, Jie and Chen, Yaoguang and Wu, Zhiyong and Fu, Jingzhou and Wang, Mingzhe and Jiang, Yu and Huang, Xiangdong and Chen, Ting and Wang, Jiashui and Li, Jiajia},
  booktitle={2023 IEEE 39th International Conference on Data Engineering (ICDE)},
  pages={668--681},
  year={2023}
}

@article{winslett2019richard,
  title={{Richard Hipp} speaks out on {SQLite}},
  author={Winslett, Marianne and Braganholo, Vanessa},
  journal={SIGMOD Rec.},
  volume={48},
  number={2},
  pages={39--46},
  year={2019}
}

\end{document}